\documentclass[a4paper]{article}
\usepackage[margin=1.0in]{geometry}

\usepackage{authblk}
\usepackage{setspace}
\usepackage{array}
\usepackage{makecell}
\usepackage{amsmath, amssymb, booktabs, mathtools}
\usepackage{subfig}
\usepackage[font=small,labelfont=bf]{caption}
\usepackage{threeparttable}
\usepackage{float}
\usepackage[square, sort&compress, numbers]{natbib}
\usepackage{bm}
\usepackage{stmaryrd}
\usepackage{multirow}
\usepackage{enumitem}

\usepackage{lipsum}
\usepackage{orcidlink}

\usepackage{amsthm}
\newtheorem{proposition}{Proposition}[section]
\newtheorem{remark}{Remark}[section]

\usepackage[ruled]{algorithm2e}
\SetKwComment{Comment}{/* }{ */}
\usepackage[hang,flushmargin]{footmisc}
\makeatletter
\newcommand{\algorithmfootnote}[2][\footnotesize]{%
  \let\old@algocf@finish\@algocf@finish  
  \def\@algocf@finish{\old@algocf@finish  
    \leavevmode\rlap{\begin{minipage}{\linewidth}
    #1#2
    \end{minipage}}%
  }%
}

\usepackage{graphicx} 
\usepackage{hyperref} 
\usepackage[nameinlink,capitalize]{cleveref}
\hypersetup{
  colorlinks   = true,
  urlcolor     = blue,
  linkcolor    = blue,
  citecolor    = blue
}
\newcommand{\correspondingA}{*}

\title{Exceptional-point conjugate symmetry and migration in fluid-loaded elastic waveguides: restructuring the physical dispersion spectrum}

\author[]{Dong Xiao\textsuperscript{*}\orcidlink{0009-0006-7609-7832}}
\author[]{Zahra Sharif Khodaei\orcidlink{0000-0001-5106-2197}}
\author[]{M.H. Aliabadi\orcidlink{0000-0002-2883-2461}} 

\affil[]{\normalsize \slshape  Department of Aeronautics, Imperial College London, South Kensington, London SW7 2AZ, United Kingdom.}
\date{\vspace*{-2\baselineskip}}  
\date{\vspace*{-1.5cm}}

\begin{document}
\maketitle
\footnotetext[0]{ \textsuperscript{\correspondingA}Corresponding author}
\footnotetext[1]{Email addresses: d.xiao21@imperial.ac.uk (D. Xiao); z.sharif-khodaei@imperial.ac.uk (Z. Sharif-Khodaei); m.h.aliabadi@imperial.ac.uk (M.H. Aliabadi)}

\renewcommand{\abstractname}{Abstract}
\begin{abstract}{\normalsize \onehalfspacing
The dispersion spectrum of a fluid-loaded elastic waveguide is a non-Hermitian system with eigenmodes on a two-sheeted Riemann manifold. Although exceptional points (EPs) organize avoided crossings, the topological rules by which fluid loading restructures the observable spectrum—both physically admissible modes and their connectivity—remain unknown, and conventional solvers fail to recover complete branches. Here we establish these rules. We prove that, for real elastic moduli and real fluid parameters, EPs obey a conjugate-pair symmetry, and that a pair on the physical sheet enforces an avoided crossing of real wavenumbers between the two modes it controls, yielding a topological criterion for mode identification. We then identify two independent reconstruction mechanisms. First, the physical observation space expands continuously: a mode observable above its vacuum cut-off can extend downward to the sound line, admitting solutions in a frequency band empty in vacuum; vacuum-anchored seeds above cut-offs continued downward capture such branches systematically. Second, EP migration rewires mode connectivity: as fluid density increases, conjugate pairs leaving the physical sheet—signaled by a sign reversal of the imaginary part of either the in-plane or the vertical wavenumber—leave the two previously coupled physical branch segments unconnected by any EP within the observation space; they become independent curves that may intersect freely on the real frequency axis, with the narrowest veerings losing their topological protection first. Mirror-symmetry breaking births additional EPs generically in the trapped regime but only conditionally in the leaky regime. Numerical computations on symmetric and asymmetric composite laminates under single- and double-sided water loading validate the framework, recovering multiple leaky branches missed by standard solvers.}
\end{abstract}

\renewcommand{\abstractname}{Keywords}
\begin{abstract}{
Fluid-loaded waveguide; Leaky Lamb waves; Scholte wave; Exceptional points; Inter-manifold transport; Non-Hermitian dispersion.}
\end{abstract}

\newpage

\section{Introduction}
\label{sec:introduction}

The interaction between an elastic waveguide and an unbounded acoustic medium constitutes one of the canonical open mechanical systems in solid mechanics.
When an elastic plate is brought into contact with a fluid, the continuous radiation degrees of freedom introduce an effective non-Hermitian self-energy that renders the guided-wave eigenvalue problem non-self-adjoint: the wavenumbers become complex, encoding radiation damping, and the spectrum acquires a nontrivial global topology. The \emph{physical} spectrum of such a system is defined by two distinct elements: its \emph{membership}---which branches are physically observable---and its \emph{topology}---how these branches connect and exchange identity along the real frequency axis. Fluid loading restructures both, and the rules governing this restructuring are the subject of the present paper.

The non-Hermitian structure of the problem is encoded in the fluid radiation condition. The vertical wavenumber $k_z = \sqrt{\omega^2/c_f^2 - k^2}$ appearing in the acoustic impedance endows the complex wavenumber plane with a natural two-sheeted structure, distinguished by the Sommerfeld condition \cite{schot_eighty_1992}: the \emph{outgoing sheet} ($\Im(k_z)>0$) carries energy away from the plate, and the \emph{incoming sheet} ($\Im(k_z)<0$) is physically inadmissible. On the real $(\omega,k)$ plane, the singular locus $k_z=0$ reduces to the \emph{sound line} $\omega = c_f k$, separating subsonic from supersonic phase velocities. On the outgoing sheet, three classes of solutions coexist: trapped modes, leaky modes, and improper modes---solutions that algebraically satisfy the dispersion relation but violate the requirement of outward energy transport.
The \emph{physical sector} is defined by the joint conditions of causality ($\Im(k)\ge 0$) and proper radiation ($\Im(k_z)>0$ with $v_p>c_f$ for leaky waves), and thus comprises only trapped and leaky modes. Although this sheet structure has been recognized for decades \cite{dayal_leaky_1989,rokhlin_topology_1989}, existing approaches to dispersion analysis treat it primarily as a \emph{numerical} inconvenience rather than a \emph{physical} object whose evolution under loading governs the observable spectrum.

Two persistent puzzles expose the absence of governing rules. The first concerns \emph{spectral membership}. Conventional solvers are known to miss physical branches systematically: higher-order leaky branches that terminate at the sound line with finite phase velocity are absent from standard outputs \cite{pavlakovic_leaky_1998,lowe_matrix_1995}; trapped Scholte waves that lack a clear vacuum counterpart may also be difficult to recover systematically. There is no internal consistency check to detect such omissions. More fundamentally, a generally accepted definition of a \emph{complete} physical spectrum has not been established for an open waveguide: if the observation space itself deforms under fluid loading, completeness relative to what reference can hardly be asserted. The second puzzle concerns \emph{spectral topology}. The physical spectrum of a fluid-loaded waveguide is non-Hermitian, and the rules that govern how its branches connect and exchange identity under loading---that is, its spectral topology---remain unknown. Numerical experiments on fluid-loaded laminates reveal that even weak loading---a small fraction of the target fluid density---can erase certain vacuum avoided crossings: the two branches, no longer coupled by an EP within the physical sector, become independent and may intersect freely on the real axis with uncorrelated imaginary parts, while wider veerings in the same spectrum survive intact. This selectivity is a symptom of the missing rules, not the puzzle itself. While perturbation theories exist for mode veering in conservative Hermitian systems \cite{xiao_mode_2026}, no framework describes the spectral topology of this open, non-Hermitian setting. These two puzzles---the expanding observation space and the unknown spectral topology---motivate the present work.

Over the past several decades, a rich methodological literature has developed to compute dispersion curves in fluid-loaded waveguides. Classical impedance--admittance formalisms provide local perturbative descriptions \cite{shuvalov_theory_2002,shuvalov_analysis_2006}, while matrix techniques---the Transfer Matrix Method \cite{thomson_transmission_1950}, the Global Matrix Method \cite{knopoff_matrix_1964}, and the Stiffness Matrix Method \cite{wang_stable_2001}---underpin widely used software packages \cite{pavlakovic_disperse_1997,huber_dispersion_2024}. The Semi-Analytical Finite Element (SAFE) method reduces the cross-sectional problem to a nonlinear eigenvalue problem for complex wavenumbers \cite{gavric_computation_1995,castaings_finite_2008}, extended to piezoelectric and multi-material waveguides \cite{bartoli_modeling_2006,marzani_semi-analytical_2008,mazzotti_coupled_2013}, and recent direct solvers exploit the analytic structure of the radiation condition to avoid spatial truncation \cite{kiefer_calculating_2019,georgiades_leaky_2022,gravenkamp_computation_2025,gravenkamp_computing_2026}. These methods differ greatly in numerics but share one limitation: they output eigenvalues at isolated frequencies without topological guidance. Branch sorting, tracing, and identification are left to post-hoc inspection, which becomes unreliable precisely where it matters most---near spectral degeneracies, and near the sound line where the observation space closes. Their systematic failures are therefore not a matter of numerical precision but of missing topological knowledge.

In parallel, exceptional-point physics has become a vibrant frontier of non-Hermitian solid mechanics. EPs have been demonstrated in PT-symmetric elastic media with balanced gain and loss \cite{rosa_exceptional_2021}, in passive elastodynamic systems with differential viscoelastic loss \cite{gupta_requisites_2023}, and in layered media through pressure--shear coupling \cite{lustig_anomalous_2019} and frozen-mode degeneracies \cite{fishman_third-order_2024}.
In open elastic systems, radiation loss provides an intrinsic non-Hermiticity: EPs have been identified in cylindrical \cite{matsushima_exceptional_2023} and spherical \cite{deguchi_observation_2025} scatterers, and dynamic encirclement of EPs has been exploited for chiral mode conversion in anisotropic elastic media \cite{duan_exceptional_2026}. A critical disconnect remains, however. These works concern engineered non-Hermiticity, scattering poles in the complex frequency plane, or EP migration induced by \emph{deliberate parameter paths} as a means of mode manipulation. By contrast, the fluid-loaded elastic waveguide is a physically intrinsic non-Hermitian system whose EPs are degeneracies of the guided-wave spectrum on the dispersion manifold---and the question of how EPs migrate along this manifold \emph{as a direct consequence of fluid loading}, and how this migration restructures the observable spectrum, has never been posed, let alone answered.

In this work, we establish the exceptional-point principles that govern the physical spectrum of fluid-loaded elastic waveguides. We first prove a conjugate-pair symmetry of exceptional points for real elastic moduli and real fluid parameters, and derive its observable signature on the real frequency axis. Building on this symmetry, we identify two independent mechanisms by which fluid loading restructures the physical spectrum: a continuous expansion of the physical observation space, and a migration of exceptional points across the admissibility boundaries of the dispersion manifold. We further analyze the conditions under which mirror-symmetry breaking creates new exceptional points. These principles are implemented in a vacuum-anchored computational framework that combines density homotopy, real-frequency continuation, and topological validation. The framework is validated on symmetric and asymmetric composite laminates under single- and double-sided water loading, including direct comparison with a widely used open-source solver. Our results recast the spectrum of an open elastic waveguide as a loading-dependent topological object, and identify the failure of conventional dispersion solvers as topological rather than numerical in origin. The present work, together with its companion studies on mode veering in conservative elastic waveguides \cite{xiao_mode_2026} and homotopy-continuation strategies for viscoelastic waveguides \cite{xiao_homotopy_2026}, forms a systematic account of the spectral topology of elastic waveguides: the conservative Hermitian limit, the methodology of inter-manifold transport, and the geometric non-Hermitian physics of fluid-loaded systems.

The remainder of this paper is organized as follows. \Cref{sec:spectral_geometry} defines the spectral geometry and the physical observation space.
\Cref{sec:EP_topology} proves the conjugate-symmetry theorem and its observable signature. \Cref{sec:observation_space,sec:migration} establish the two restructuring mechanisms. \Cref{sec:computation} presents the computational framework, \cref{sec:validation} the laminate validations, and \cref{sec:conclusion} concludes.

\section{Spectral geometry of the fluid-loaded elastic waveguide}
\label{sec:spectral_geometry}

This section establishes the geometric stage on which the entire analysis takes place.
We first recall the governing equations and the origin of non-Hermiticity (\cref{sec:governing}), then define the physical sector of the dispersion manifold through the admissibility conditions (\cref{sec:outgoing_sheet}), and finally introduce the \emph{physical observation space}---the projection of the admissible spectrum onto the experimentally accessible plane---whose deformation under fluid loading (\cref{sec:observation_space}) and whose internal connectivity (\cref{sec:migration}) constitute the two restructuring mechanisms studied in this work.

\subsection{Governing equations and geometric non-Hermiticity}
\label{sec:governing}

Consider an elastic plate of infinite extent in the $x$--$y$ plane and finite thickness $h$ along $z$, composed of homogeneous anisotropic laminas with real symmetric stiffness tensors $\mathbf{C}$ (purely elastic response; viscous damping is deliberately excluded). The plate is coupled on its upper and/or lower surface to a stationary, inviscid, compressible fluid of density $\rho_f$ and sound speed $c_f$.

The solid obeys the elastodynamic equation
\begin{equation}
    \nabla \cdot \boldsymbol{\sigma} = \rho_s \ddot{\mathbf{u}},
    \qquad
    \boldsymbol{\sigma} = \mathbf{C} : \boldsymbol{\varepsilon},
    \label{eq:elastodynamic}
\end{equation}
with $\rho_s$ the solid density. In the fluid half-space the acoustic pressure satisfies the Helmholtz equation
\begin{equation}
    \nabla^2 p_f + k_f^2 p_f = 0,
    \qquad
    k_f = \omega / c_f.
    \label{eq:helmholtz}
\end{equation}
At the fluid--solid interface the normal traction and normal displacement are continuous:
\begin{equation}
    \boldsymbol{\sigma}\cdot\mathbf{n} = -p_f \mathbf{n},
    \qquad
    \mathbf{u}\cdot\mathbf{n} = \mathbf{u}_f\cdot\mathbf{n}.
    \label{eq:interface}
\end{equation}

Seeking time-harmonic guided waves propagating in the $x$ direction,
\begin{equation}
    \{\mathbf{u}, p_f\}(x,z,t) = \{\hat{\mathbf{u}}(z), \hat{p}_f(z)\}\, e^{i(kx-\omega t)},
    \label{eq:ansatz}
\end{equation}
reduces the problem to a one-dimensional eigenvalue problem in the transverse coordinate $z$. The fluid pressure that satisfies the radiation condition at infinity is
\begin{equation}
    \hat{p}_f(z) = P_0\, e^{i k_z (z-z_{\rm surf})},
    \qquad
    k_z = \sqrt{k_f^2 - k^2},
    \label{eq:kz_def}
\end{equation}
where the branch of the square root is chosen such that $\Im(k_z) > 0$ (outgoing or evanescent decay), and when $\Im(k_z)=0$, $\Re(k_z)\ge 0$. The normal component of the linearized Euler equation together with the kinematic interface condition \cref{eq:interface} yields the acoustic impedance
\begin{equation}
    Z(\omega,k) = -\frac{i\rho_f \omega^2}{k_z},
    \label{eq:impedance}
\end{equation}
which enters the structural boundary condition as a non-local, frequency- and wavenumber-dependent stiffness. The resulting dispersion relation is a \emph{nonlinear} eigenvalue problem
\begin{equation}
    \mathbf{D}(\omega,k)\,\mathbf{q} = \mathbf{0},
    \label{eq:dispersion}
\end{equation}
where $\mathbf{q}$ collects the modal degrees of freedom and $\mathbf{D}$ is the dynamic matrix. For the vacuum plate ($\rho_f=0$), $\mathbf{D}$ is Hermitian for real $k$; the spectrum is real, eigenmodes are orthogonal, and level repulsion forbids true crossings between modes of the same symmetry family. The fluid introduces a geometric non-Hermitian self-energy through $Z(\omega,k)$, rendering $\mathbf{D}$ non-self-adjoint and the eigenvalues $k$ complex.

In what follows, the specific discretization used to obtain $\mathbf{D}$ (here, a Semi-Analytical Finite Element scheme with spectral elements \cite{gavric_computation_1995,castaings_finite_2008}) is immaterial to the topological results; the conclusions hold for any consistent spatial discretization of the continuum problem \cref{eq:elastodynamic}--\cref{eq:interface}. The explicit spectral-element construction of $\mathbf{D}(\omega,k)$, including the definitions of $\mathbf{K}_1$, $\mathbf{K}_2$, $\mathbf{K}_3$, $\mathbf{M}$, and the fluid-loading projection $\mathbf{H}$, is given in \cref{app:SAFE}. All topological results of \cref{sec:EP_topology,sec:observation_space,sec:migration} are independent of the specific choice of basis functions, provided the discretization is consistent.

\subsection{Radiation sheets, admissibility conditions, and the physical sector}
\label{sec:outgoing_sheet}

The square-root function $k_z(\omega,k)$ in \cref{eq:kz_def} endows the complex wavenumber plane with a natural two-sheeted structure. The two \emph{radiation sheets} are distinguished by the Sommerfeld condition \cite{schot_eighty_1992}: the \emph{outgoing sheet} $\mathcal{R}_{\rm out}$ ($\Im(k_z)>0$), on which the acoustic field carries energy away from the plate or decays exponentially, and the \emph{incoming sheet} $\mathcal{R}_{\rm in}$ ($\Im(k_z)<0$), which is physically inadmissible. The sheets meet along the singular locus
\begin{equation}
    \mathcal{S} = \{(\omega,k) : k_z(\omega,k)=0\},
    \label{eq:singular_locus}
\end{equation}
where the fluid impedance \cref{eq:impedance} diverges. On the real $(\omega,k)$ plane, $\mathcal{S}$ reduces to the \emph{sound line} $\omega = c_f k$, marking the transition between subsonic and supersonic phase velocities. It is essential for what follows that $\mathcal{S}$ is a singularity of the \emph{operator} $\mathbf{D}$ where the fluid impedance diverges. The dispersion manifold does not cross $\mathcal{S}$, and eigenpairs exist only on one side of it. Trapped modes whose wavenumbers are real may approach $\mathcal{S}$ as $k\to\omega/c_f$, where the impedance singularity terminates the branch. Leaky modes, by contrast, remain at a finite distance from $\mathcal{S}$ even when their phase velocity approaches $c_f$, because their complex wavenumbers keep $|k_z|$ bounded away from zero. The locus $\mathcal{S}$ therefore acts as a global constraint on the parameter domain, but the way a branch interacts with it depends on the branch type.

Physical solutions are selected by two independent admissibility conditions:
\begin{enumerate}
    \item[(i)] \emph{Causality:} $\Im(k) \ge 0$, so that the wave either propagates without growth or decays along the guide \cite{rokhlin_topology_1989}.
    \item[(ii)] \emph{Radiation:} $\Im(k_z) > 0$, selecting the outgoing sheet.
\end{enumerate}
The real hypersurfaces $\Im(k)=0$ and $\Im(k_z)=0$ in the complex parameter space are therefore \emph{admissibility boundaries}: across them, the exponential behavior of the wave changes character (growth versus decay along the guide; incoming versus outgoing in the fluid). These boundaries are physically meaningful loci---not conventions---and they will play a central role in \cref{sec:migration}, where exceptional points are shown to cross them under fluid loading. The intersection of the two admissible half-spaces defines the \emph{physical sector} of the dispersion manifold.

Solutions on the outgoing sheet that satisfy $\Im(k)\ge 0$ fall into three disjoint classes according to their phase velocity and radiation character:

\begin{itemize}
    \item \textbf{Trapped modes (Scholte-type).} $\Im(k)=0$, $|\Re(k)| > \omega/c_f$ (subsonic, $v_{\rm p}<c_f$). Here $k_z$ is purely imaginary with $\Im(k_z)>0$; the fluid field decays exponentially, and the impedance is real (pure added mass). These modes carry no radiation loss.

    \item \textbf{Leaky modes (leaky Lamb waves).} $\Im(k)>0$, $|\Re(k)| < \omega/c_f$ (supersonic, $v_{\rm p}>c_f$). Here $k_z$ is complex with $\Im(k_z)>0$; the fluid field radiates energy to infinity, and the positive $\Im(k)$ encodes attenuation along the guide.

    \item \textbf{Improper modes.} $\Im(k)>0$, $\Im(k_z)>0$, yet $|\Re(k)| > \omega/c_f$ (subsonic). Although they satisfy the algebraic dispersion relation, their phase velocity is below the fluid sound speed, so no propagating wave is launched into the fluid; they are non-physical artifacts frequently emitted by fixed-frequency root solvers.
\end{itemize}

A solution is physically admissible if and only if it belongs to the trapped or leaky categories. All solutions on $\mathcal{R}_{\rm in}$ violate the radiation condition and are excluded.

\begin{remark}[Two distinct singular structures]
\label{rem:two_singularities}
The problem contains two singular structures of entirely different nature, which must not be conflated. The first is the operator singularity $\mathcal{S}$ \cref{eq:singular_locus}: a surface through which the dispersion manifold never passes, acting as a global constraint on the parameter domain. The second consists of the \emph{exceptional points}: branch points of the eigenvalue function on the dispersion manifold itself, located at finite distance from $\mathcal{S}$, where two eigenvalues \emph{and} their eigenvectors coalesce (rendering $\mathbf{D}$ defective). The former governs the global boundedness of the spectral landscape; the latter governs the local connectivity of branches. The restructuring mechanisms of \cref{sec:observation_space,sec:migration} are organized by the interplay of the two: the observation space is bounded by $\mathcal{S}$, while its internal topology is controlled by exceptional points.
\end{remark}

\subsection{The dispersion manifold and the physical observation space}
\label{sec:obs_space}

For the analysis of exceptional-point migration in \cref{sec:migration}, it is convenient to regard $k_z$ as an independent complex variable constrained by the fluid dispersion relation. The \emph{dispersion manifold} is then
\begin{equation}
    \mathcal{M} = \left\{ (\omega,k,k_z)\in\mathbb{C}^3 : \det\mathbf{D}(\omega,k,k_z)=0, \;\; k_z^2 + k^2 = \frac{\omega^2}{c_f^2} \right\},
    \label{eq:dispersion_manifold}
\end{equation}
a multi-sheeted branched covering of the parameter space whose sheets are the eigenvalue branches. Its branch points are the exceptional points studied in \cref{sec:EP_topology}. The admissibility conditions of \cref{sec:outgoing_sheet} carve out of $\mathcal{M}$ the physical sector
\begin{equation}
    \mathcal{M}_{\rm phys} = \left\{ (\omega,k,k_z)\in\mathcal{M} : \Im(k)\ge 0,\ \Im(k_z)>0 \right\},
    \label{eq:physical_sector}
\end{equation}
whose boundary segments lie on the admissibility hypersurfaces $\Im(k)=0$ and $\Im(k_z)=0$.

Experiments and standard computations do not access $\mathcal{M}$ directly; they access its projection onto the real frequency--wavenumber plane. We define the \emph{physical observation space} as
\begin{equation}
    \Omega_{\rm obs} = \left\{ (\omega,\Re k)\in\mathbb{R}^2 : (\omega,k,k_z)\in\mathcal{M}_{\rm phys},\ \omega\in\mathbb{R} \right\},
    \label{eq:observation_space}
\end{equation}
together with the attenuation $\Im(k)$ and energy velocity $v_e$ attached to each point. The set $\Omega_{\rm obs}$ is, by construction, unbounded in frequency: it contains every physically admissible guided‑wave solution of the plate–fluid system. In practice, however, one interrogates only a finite \emph{observation window} dictated by the application. For ultrasonic nondestructive evaluation and structural health monitoring, the frequency–thickness product rarely exceeds a few $\mathrm{MHz \cdot mm}$, and all numerical illustrations in this work are confined to this window. The unbounded theoretical definition remains essential, because it allows us to describe the loading‑induced extension of branches down to the sound line without artificial truncation---an extension that may bring physically admissible solutions into the observation window from frequency intervals that were empty in the vacuum limit.

Two facets of $\Omega_{\rm obs}$ organise this paper. Its \emph{membership}---which branches populate it, and over which frequency intervals---is set by the admissibility conditions and by the sound line. Its \emph{topology}---the pattern of veerings and crossings through which branches interact and exchange identity---is set by the exceptional points of $\mathcal{M}$ that reside in the physical sector.

The structure of $\Omega_{\rm obs}$ in the vacuum limit is instructive. Each Lamb mode of the free plate is observable above its cut-off frequency $\omega_{\rm c}$, so that its contribution to $\Omega_{\rm obs}$ is a branch supported on the half-line $[\omega_{\rm c},\infty)$; below $\omega_{\rm c}$ the mode possesses no real-wavenumber solution and simply does not exist in the observation space. The vacuum observation space is thus a union of half-line-supported branches, partitioned by the sound line into subsonic and supersonic segments. One of the central findings of this work (\cref{sec:observation_space}) is that fluid loading deforms this structure continuously: the lower endpoint of a branch can migrate from the vacuum cut-off down to the sound line, so that solutions without any real-$k$ vacuum existence enter the observation space---while, simultaneously, the connectivity of branches within the observation space is rewritten by exceptional-point migration (\cref{sec:migration}). The physical spectrum is therefore not a fixed object inherited from the vacuum plate, but a loading-dependent section of the dispersion manifold.

\section{Exceptional-point symmetry and its observable signature}
\label{sec:EP_topology}

This section establishes the \emph{static} rules that exceptional points (EPs) obey on the dispersion manifold $\mathcal{M}$: the conjugate-pair symmetry inherited from the realness of the constitutive parameters (\cref{sec:EP_symmetry}); the signature that a physical-sector EP pair imprints on real-frequency mode interactions, which serves as the topological criterion for mode identification (\cref{sec:signature}); and the classification of the distinct ways an EP can cross the admissibility boundaries under parameter variation, setting the stage for the migration analysis of \cref{sec:migration}.

\subsection{Conjugate symmetry of exceptional-point pairs}
\label{sec:EP_symmetry}

The dynamic matrix $\mathbf{D}(\omega,k,k_z)$ of the augmented formulation \cref{eq:dispersion_manifold} inherits a fundamental symmetry from the realness of the elastic moduli and fluid parameters.

\begin{proposition}[Conjugate symmetry of the dynamic matrix]
\label{prop:Dconjugate}
For real elastic moduli and real fluid parameters, the dynamic matrix satisfies
\begin{equation}
    \mathbf{D}(\omega^*,-k^*,-k_z^*) = \mathbf{D}(\omega,k,k_z)^*
    \label{eq:D_conjugate}
\end{equation}
for all $(\omega,k,k_z)$ satisfying the constraint $k_z^2+k^2=\omega^2/c_f^2$, i.e.\ on the entire dispersion manifold $\mathcal{M}$.
\end{proposition}

\begin{proof}
The structural stiffness and mass matrices are real and symmetric; the gyroscopic matrix (linear in $k$) is purely imaginary and skew-symmetric, so that $(i k \mathbf{K}_2)^* = -i(-k^*)\mathbf{K}_2$. For the fluid term, $Z\propto\omega^2/k_z$ with real $\rho_f$, $c_f$ gives $Z(\omega^*,-k_z^*) = Z(\omega,k_z)^*$. Substitution into $\mathbf{D}$ yields \cref{eq:D_conjugate}. The constraint is preserved: $(-k_z^*)^2+(-k^*)^2-(\omega^*)^2/c_f^2 = (k_z^2+k^2-\omega^2/c_f^2)^* = 0$.
\end{proof}

An exceptional point $(\omega_{\mathrm{EP}},k_{\mathrm{EP}},k_{z,\mathrm{EP}})\in\mathcal{M}$ is a spectral degeneracy where two eigenvalues and their eigenvectors coalesce, rendering $\mathbf{D}$ defective \cite{heiss_physics_2012}.

\begin{proposition}[Conjugate symmetry maps forward EPs to backward EPs]
\label{prop:EPpair}
If $(\omega_{\mathrm{EP}},k_{\mathrm{EP}},k_{z,\mathrm{EP}})$ is an exceptional point, then so is
\begin{equation}
    (\omega_{\mathrm{EP}}^*,\,-k_{\mathrm{EP}}^*,\,-k_{z,\mathrm{EP}}^*).
    \label{eq:EP_conjugate}
\end{equation}
The map $(\omega,k,k_z)\mapsto(\omega^*,-k^*,-k_z^*)$ sends an exceptional point with wavenumber $k_{\mathrm{EP}}$ to one with wavenumber $-k_{\mathrm{EP}}^*$. Consequently, if we designate a \emph{forward} sector by $\Re k>0$, then the image EP lies in the \emph{backward} sector ($\Re k<0$). The real parts of the two wavenumbers are opposite in sign, while the imaginary parts of both $k$ and $k_z$ are preserved:
\begin{equation}
\Re(-k_{\mathrm{EP}}^*) = -\Re k_{\mathrm{EP}},
\qquad
\Im(-k_{\mathrm{EP}}^*) = \Im k_{\mathrm{EP}},
\qquad
\Im(-k_{z,\mathrm{EP}}^*) = \Im k_{z,\mathrm{EP}}.
\end{equation}
The map therefore acts within each admissibility region of $\mathcal{M}$, and in particular preserves the physical sector $\mathcal{M}_{\mathrm{phys}}$.
\end{proposition}

\begin{proof}
From \cref{prop:Dconjugate}, $\det\mathbf{D}(\omega_{\mathrm{EP}}^*,-k_{\mathrm{EP}}^*,-k_{z,\mathrm{EP}}^*) = [\det\mathbf{D}(\omega_{\mathrm{EP}},k_{\mathrm{EP}},k_{z,\mathrm{EP}})]^* = 0$, and the degeneracy condition transforms identically, preserving the coalescence order. The sign preservation follows from $\Im(-k^*)=\Im(k)$ and $\Im(-k_z^*)=\Im(k_z)$.
\end{proof}

\begin{remark}[No self-conjugate EP; projection of the pair in the Hermitian limit]
\label{rem:pairing_scope}
The conjugate symmetry established in \cref{prop:EPpair} also guarantees that a forward EP and its backward image never coincide. A self-conjugate point would require $\omega_{\mathrm{EP}}=\omega_{\mathrm{EP}}^*$ (real $\omega$), $k_{\mathrm{EP}}=-k_{\mathrm{EP}}^*$ (pure imaginary $k$), and $k_{z,\mathrm{EP}}=-k_{z,\mathrm{EP}}^*$ (pure imaginary $k_z$). At real frequencies the elastodynamic operator is either Hermitian (trapped regime, real $k$) or genuinely non-Hermitian (leaky regime, complex $k$ with nonzero real part); a pure imaginary wavenumber $k$ together with a pure imaginary $k_z$ cannot satisfy the dispersion relation and the fluid constraint simultaneously at any finite $(\omega,k,k_z)$. Hence no self-conjugate EP exists on $\mathcal{M}$. Even when an EP reaches the real frequency axis ($\Im\omega_{\mathrm{EP}}=0$), it appears as two distinct points at the same frequency in the forward and backward sectors, corresponding to coalescence in the two opposite propagation directions.
\end{remark}

\subsection{The observable signature of a physical-sector EP pair}
\label{sec:signature}

Consider two interacting forward-propagating physical branches
$k_A(\omega)$ and $k_B(\omega)$ traced along the real frequency axis.
Their interaction is controlled by a single conjugate pair of
exceptional points, one in the forward sector and one in the backward
sector:
\begin{equation}
P_{\mathrm f}=(\omega_{\mathrm{EP}}, k_{\mathrm{EP}}, k_{z,\mathrm{EP}}),
\qquad
P_{\mathrm b}=(\omega_{\mathrm{EP}}^*,-k_{\mathrm{EP}}^*,-k_{z,\mathrm{EP}}^*).
\label{eq:conjugate_EP_pair}
\end{equation}
As shown in \cref{prop:EPpair}, the two partners have the same
frequency real part $\Re\omega_{\mathrm{EP}}$, opposite real
wavenumber parts $\pm\Re k_{\mathrm{EP}}$, and identical imaginary
parts of both $k$ and $k_z$. If $\Im\omega_{\mathrm{EP}}\neq0$, the
two EPs lie on opposite sides of the real frequency axis.

The observable signature depends on whether the two interacting
branches belong to the trapped or to the leaky regime.

\paragraph{Trapped regime.}
In the trapped regime the admissible wavenumbers are real ($\Im k=0$), and the fluid impedance is real. The eigenvalue problem on the real
frequency axis is therefore Hermitian, and the two forward branches
have real wavenumbers. Since the real axis contains no exceptional
point, the two branches remain distinct for all real frequencies.
Two real branches can be distinct only through separation in their
real parts; they therefore exhibit the familiar real-wavenumber
veering,
\begin{equation}
\boxed{\Re k:\ \text{veering},\qquad \Im k=0.}
\label{eq:trapped_signature}
\end{equation}
This veering is precisely the projection of the conjugate EP pair
\cref{eq:conjugate_EP_pair} onto the real axis: the EPs lie off the
real frequency axis, one above and one below, and the real axis passes
between them.

\paragraph{Leaky regime.}
In the leaky regime the admissible wavenumbers are genuinely complex,
$k = k_{\mathrm r} + i k_{\mathrm i}, \; k_{\mathrm i}\neq0$,
and the system is non-Hermitian. The
interaction of the two forward branches is governed by the same
conjugate EP pair. Here we assume that both EPs satisfy the physical
admissibility conditions, so that they reside within the physical sector of the dispersion manifold ($ P_{\mathrm f}, P_{\mathrm b} \in \mathcal{M}_{\rm phys}$) and exercise a direct topological control on the physical dispersion curves.

Two cases must be distinguished.

\emph{EPs on the real frequency axis.} If
$\Im\omega_{\mathrm{EP}}=0$, the two conjugate EPs coincide in
frequency, lying at $(\omega_0,k_{\mathrm{EP}},k_{z,\mathrm{EP}})$ and
$(\omega_0,-k_{\mathrm{EP}}^*,-k_{z,\mathrm{EP}}^*)$. Their
wavenumbers have opposite real parts and equal imaginary parts, so the
two EPs are distinct points of $\mathcal{M}$ but share the same real
frequency. In this degenerate situation the real-axis scan passes
through the exceptional points themselves, and the two physical
branches cannot be assigned a unique continuous mode identity across
the interaction. This case requires fine tuning and is not generic; it
is excluded from the remainder of the discussion.

\emph{EPs off the real frequency axis.} In the generic case
$\Im\omega_{\mathrm{EP}}\neq0$, the real frequency axis contains no
exceptional point, and the two forward branches
$k_A(\omega), k_B(\omega)$ are single-valued analytic functions on the
entire real line. The EP topology still enforces the exchange of the
two modal characters when the real axis is traversed through the
interaction region: a contour consisting of the real interval
$[\omega_a,\omega_b]$ straddling the interaction and a large semicircle
in the upper half-plane encloses exactly one EP, and analytic
continuation around that loop interchanges the two branches. As a
result, the physical modes exchange their identity between the low-
and high-frequency sides of the interaction.

Because the two EPs lie on opposite sides of the real axis and both
reside in the physical sector, this exchange is realized as a veering
of the real wavenumbers together with a crossing of the imaginary
wavenumbers. The real parts remain separated throughout the
interaction; the imaginary parts, which are identical at the level of
the analytic sheets, are interchanged when the branches are followed
according to a continuous physical label such as attenuation or
energy velocity. The observable signature in the leaky regime is
therefore
\begin{equation}
\boxed{\Re k:\ \text{veering},\qquad \Im k:\ \text{crossing}.}
\label{eq:leaky_signature}
\end{equation}
The imaginary-part crossing is a physical-label crossing: the two
complex eigenvalues remain distinct at the crossing point,
\begin{equation}
\Im k_A=\Im k_B, \qquad \Re k_A\neq\Re k_B,
\end{equation}
and hence $k_A\neq k_B$.

The minimum separation between the two branches is controlled by the
distance of the conjugate EP pair from the real frequency axis. Based on the local Puiseux expansion \cite{nennig_high_2020}, a local two-EP description gives 
\begin{equation}
|k_A-k_B|
\simeq
2|c_2|
\sqrt{(\omega-\Re \mathrm{EP})^2+\Im \mathrm{EP}^2},
\end{equation}
so that
\begin{equation}
\min_{\omega\in\mathbb{R}}
|k_A-k_B|
\sim
2|c_2|\,|\Im \mathrm{EP}|.
\label{eq:gap_EP_distance}
\end{equation}
Thus the veering gap scales linearly with $|\Im \mathrm{EP}|$, the distance of
the EP pair from the real axis.

\begin{proposition}[Observable signature of a physical-sector EP pair]
\label{prop:signature}
Consider two interacting forward-propagating physical modes whose
interaction is controlled by a conjugate pair of physical-sector EPs
\eqref{eq:conjugate_EP_pair} with $\Im\omega_{\mathrm{EP}}\neq0$.
Then along the real frequency axis:
\begin{enumerate}[label=(\roman*)]
    \item in the trapped regime, the two branches exhibit
    real-wavenumber veering with zero imaginary wavenumber;
    \item in the leaky regime, the two branches exhibit
    real-wavenumber veering and imaginary-wavenumber crossing.
\end{enumerate}
The latter statement applies only to branches identified by a
continuous physical label; the analytic Riemann sheets remain distinct
throughout the interaction.
\end{proposition}

\begin{remark}[Diagnostic interpretation and static character]
\label{rem:diagnostic}
These rules are \emph{static}: they hold for a fixed loading and a given location of the controlling exceptional-point pair. The signature established in Proposition~\ref{prop:signature} applies to mode pairs that are continuously connected to a vacuum avoided crossing and whose conjugate EP pair resides in the physical sector. For such pairs, a genuine real-wavenumber crossing with uncorrelated imaginary parts signals that the EP-mediated interaction has been lost from the physical sector. Mode pairs without a vacuum avoided crossing may cross freely without any EP involvement. The loading-induced migration that can drive a physical-sector EP pair out of the physical sector is the subject of \cref{sec:migration}.
\end{remark}

\section{Reconstruction of the physical observation space under fluid loading}
\label{sec:observation_space}

The vacuum spectrum of an elastic plate provides a complete, real‑wavenumber reference: every guided mode with a non‑zero cut‑off frequency $\omega_{\rm c}>0$ exists only on the half‑line $[\omega_{\rm c},\infty)$, where its wavenumber is real. At $\omega=\omega_{\rm c}$ the wavenumber vanishes, the phase velocity diverges, and the mode is therefore always supersonic relative to any finite fluid sound speed $c_f$. Below $\omega_{\rm c}$ the vacuum dispersion equation admits no real‑wavenumber solution; all roots lie in the complex $k$‑plane and represent non‑propagating motions that do not belong to the physical observation space $\Omega_{\rm obs}$.

In this section we establish the first mechanism by which fluid loading restructures the observable spectrum: the continuous \emph{expansion} of the physically admissible frequency interval of a mode. We show that fluid loading turns the evanescent complex roots below $\omega_{\rm c}$ into leaky modes that satisfy the radiation condition, thereby extending the branch continuously downward from the cut‑off. The extension terminates at the sound line $\omega=c_f\Re k$, where the phase velocity drops below $c_f$ and the mode ceases to be a legitimate leaky solution. This mechanism explains why conventional solvers, which initialize all branches at infinite phase velocity, systematically miss the sound‑line‑terminating segments that populate the newly opened frequency window.

\subsection{Deformation of the admissible frequency interval}
\label{sec:interval_deformation}

Consider a vacuum Lamb mode with cut‑off frequency $\omega_{\rm c}>0$. In vacuum the mode contributes to $\Omega_{\rm obs}$ a branch supported on $[\omega_{\rm c},\infty)$, with $k=0$ at $\omega_{\rm c}$ (infinite phase velocity) and $v_{\rm p}>c_f$ at least in a neighborhood above $\omega_{\rm c}$. Below $\omega_{\rm c}$ all vacuum roots of $\det\mathbf{D}_0(\omega,k)=0$ are complex; they describe non‑propagating fields and are absent from $\Omega_{\rm obs}$.

When fluid loading is switched on, the impedance $Z(\omega,k_z)$ perturbs the eigenvalue problem and converts the evanescent vacuum roots below $\omega_{\rm c}$ into complex‑wavenumber solutions on the outgoing sheet. For a small frequency interval immediately below $\omega_{\rm c}$, the phase velocity remains close to infinity, so these solutions are supersonic and automatically satisfy the leaky‑mode condition $\Im(k)>0$, $\Im(k_z)>0$. They therefore belong to the same analytic branch as the original vacuum mode and extend it into the previously empty sub‑cut‑off region.

As frequency decreases further, the complex wavenumber $k(\omega)$ evolves continuously: $\Re k$ increases so that the phase velocity $v_{\rm p}=\omega/\Re k$ falls towards $c_f$, while $\Im k$ remains positive and typically grows, indicating increasing attenuation. Throughout this process $k$ stays complex, $k_z$ remains complex with $\Im(k_z)>0$, and the branch does \emph{not} approach the singular locus $k_z=0$---the impedance stays finite. The extension is therefore a genuine motion of the eigenvalue on the regular part of the outgoing sheet, driven by the fluid perturbation.

The lower endpoint of the extended branch in $\Omega_{\rm obs}$ is set by the sound line: at some frequency $\omega_*$ the condition $v_{\rm p}=c_f$ is reached. Below $\omega_*$ the phase velocity becomes subsonic ($v_{\rm p}<c_f$); the mode can no longer sustain a propagating acoustic field in the fluid and turns into an improper solution, which is excluded from the physical spectrum. Hence the physically observable segment of the branch terminates at the sound line, and the branch now populates the interval $[\omega_*,\infty)$ rather than $[\omega_{\rm c},\infty)$.

The extent of the extension depends on the fluid sound speed $c_f$. For a given mode, a larger $c_f$ shifts the sound line to higher phase velocities, shortening the frequency interval over which the extended branch remains supersonic; consequently, only higher‑order modes---whose cut‑off frequencies are larger and whose phase velocities fall more slowly with decreasing $\omega$---may reach the sound line when the fluid is fast. This parametric dependence is consistent with the numerical examples of \cref{sec:validation}, where multiple higher‑order leaky branches terminate at the sound line under water loading.

\begin{proposition}[Extension of a vacuum cut‑off branch to the sound line]
\label{prop:extension}
Let a vacuum Lamb mode have cut‑off frequency $\omega_{\rm c}>0$. Under fluid loading, the branch deforms continuously: its observable frequency interval expands from $[\omega_{\rm c},\infty)$ to $[\omega_*,\infty)$, where $\omega_*$ is the frequency at which $v_{\rm p}=c_f$. The added segment $[\omega_*,\omega_{\rm c})$ consists of leaky modes with $\Im(k)>0$ and $\Im(k_z)>0$; the branch remains at a finite distance from the singular locus $k_z=0$ throughout this segment. Below $\omega_*$ the solution becomes improper and is excluded from $\Omega_{\rm obs}$.
\end{proposition}

\subsection{Systematic recovery by vacuum‑anchored seed selection and continuation}
\label{sec:downward_continuation}

Conventional fixed‑frequency root solvers initialize their search at cut‑off frequencies assuming $k=0$ or $v_{\rm p}\to\infty$; they have no mechanism to discover solutions that terminate at finite phase velocity. Our framework instead anchors every branch in the complete, Hermitian vacuum spectrum and classifies each vacuum mode by its phase velocity relative to the fluid sound speed $c_f$ within the observation window $[0,\omega_{\max}]$:

\begin{enumerate}[label=(\roman*)]
    \item \textbf{Entirely supersonic:} $v_{\rm p}(\omega)>c_f$ for all $\omega\in[\omega_{\rm c},\omega_{\max}]$. The mode is leaky at all frequencies in the window. A single seed is placed at $\omega_{\max}$.
    \item \textbf{Entirely subsonic:} $v_{\rm p}(\omega)<c_f$ for all $\omega\in[\omega_{\rm c},\omega_{\max}]$. The mode is trapped at all frequencies in the window. A single seed is placed at $\omega_{\max}$.
    \item \textbf{Straddling the sound line:} the vacuum branch crosses $\omega=c_f k$ at some frequency $\omega_{\rm cross}\in(\omega_{\rm c},\omega_{\max})$. One vacuum mode produces \emph{two} physically distinct fluid‑loaded branches. Two seeds are required: one placed just above $\omega_{\rm c}$ (to capture the downward extension toward the sound line), and one at $\omega_{\max}$ (to capture the supersonic continuation).
\end{enumerate}

Seeds in classes (i) and (ii) are placed at $\omega_{\max}$---the upper end of the observation window---because the vacuum eigenvalue there is well separated from neighboring branches and far from both the cut‑off and the sound line, providing a robust starting point for homotopy. For straddling modes (class iii), the seed near $\omega_{\rm c}$ is essential: it anchors the sub‑cut‑off extension that terminates at the sound line, which would be inaccessible from a seed placed solely at high frequency.

With seeds selected, the inter‑manifold transport proceeds as follows:
\begin{enumerate}[label=(\arabic*)]
    \item \textbf{Density homotopy.} The fluid density is ramped from $0$ to $\rho_f$ along a complex path in the homotopy parameter $s$, transporting each vacuum seed $(k_{\rm vac},\mathbf{q}_{\rm vac})$ onto the fluid‑loaded manifold while preserving analytic branch identity (\cref{app:homotopy}).
    \item \textbf{Frequency continuation.} Starting from each seed at $s=1$, the augmented system \cref{eq:augmented_app} is solved with $\omega$ as the continuation parameter. For seeds at $\omega_{\max}$, continuation proceeds downward toward the cut‑off and, if the branch so permits, beyond it toward the sound line. For the low‑frequency seed of a straddling mode, continuation proceeds both upward (to connect with the supersonic segment) and downward (to track the extension to the sound line). The step size is controlled by the MAC criterion \cref{eq:MAC_app}. Continuation terminates when:
    \begin{enumerate}
        \item the branch reaches the sound line ($v_{\rm p}=c_f$), below which the phase velocity becomes subsonic and the solution turns improper; or
        \item $\Im(k)$ changes sign, signaling exit from the physical sector.
    \end{enumerate}
\end{enumerate}

This strategy mirrors the physics of \cref{prop:extension}: the observable branch is accessed from its well‑characterized Hermitian anchor, and the loading‑induced extension below $\omega_{\rm c}$ is recovered automatically, without \emph{a priori} knowledge of the sound‑line intersection.

\begin{remark}[What this procedure does not guarantee]
\label{rem:incompleteness}
The procedure recovers every branch whose vacuum ancestor has a cut‑off frequency within the seeded range. Modes whose cut‑offs lie above $\omega_{\max}$, and which might also extend to the sound line, are not captured. The physical spectrum of an open waveguide is infinite‑dimensional; no finite‑frequency search can claim exhaustiveness. What is guaranteed is that every branch traced carries a topologically validated identity and that, within the seeded frequency band, the membership of $\Omega_{\rm obs}$ is complete \emph{relative to the traced vacuum ancestors}.
\end{remark}

\subsection{Relation to Scholte waves}
\label{sec:scholte_relation}

The vacuum classification also predicts when trapped Scholte waves emerge without a vacuum counterpart. A mode of class (ii) (entirely subsonic in the observation window) corresponds under fluid loading to a trapped branch that remains subsonic down to arbitrarily low frequencies. If the fundamental symmetric mode S$_0$ satisfies $v_{\rm p}(\omega)<c_f$ throughout the window, then under double‑sided loading it gives rise to the familiar symmetric Scholte wave---a purely bound interface mode with $v_{\rm p}\lesssim c_f$ at low frequencies. This Scholte wave \emph{does} possess a vacuum ancestor (the S$_0$ mode) and is recovered by the seed‑and‑continuation procedure without any special treatment.

The only case requiring a direct search is when the vacuum S$_0$ mode is entirely supersonic ($v_{\rm p}>c_f$ for all $\omega$) \emph{and} loading is double‑sided. In this situation, which can occur for very soft solids or very light fluids, the symmetric Scholte wave has no vacuum counterpart. It is an emergent trapped mode, located by a direct low‑frequency root search on the fully loaded system, as described in \cref{sec:emergent_computation}. For the composite laminates and water loading considered in \cref{sec:validation}, the vacuum S$_0$ mode is subsonic over most of the observation window, so the symmetric Scholte wave is captured naturally by the seed‑and‑continuation framework; only the antisymmetric Scholte wave under single‑sided loading requires the emergent‑mode treatment.

\subsection{Observable consequence: branches systematically missed by conventional solvers}
\label{sec:missing_branches}

The standard approach to computing fluid‑loaded dispersion—a fixed‑frequency root search initialized with $k=0$—implicitly assumes that every mode starts its existence at infinite phase velocity. The analysis of \cref{sec:interval_deformation} shows that this assumption is false: an entire class of physical branches originates at the sound line, at finite phase velocity, and possesses no real‑wavenumber vacuum counterpart in the frequency interval where they appear. Consequently, solvers that rely on vacuum cut‑off initialization skip these branches without any diagnostic.

In \cref{sec:validation} we demonstrate the recovery of multiple sound‑line‑terminating branches for both symmetric and asymmetric composite laminates under water loading. Direct comparison with a widely used open‑source solver confirms that these branches are absent from its output, corroborating the topological, rather than numerical, origin of the omission.

\subsection{Connection to the second restructuring mechanism}
\label{sec:connection_to_EP}

The expansion of $\Omega_{\rm obs}$ changes \emph{which} branches are visible; it does not alter \emph{how} they connect. The internal topology of the observation space—the pattern of veerings and crossings among the now‑extended branches—remains governed by the exceptional points of the dispersion manifold. The migration of those EPs across the admissibility boundaries, which rewires branch connectivity without changing branch membership, is the subject of \cref{sec:migration}. Together, the two mechanisms constitute a complete picture of the restructuring of the physical spectrum of a fluid‑loaded waveguide.

\section{Riemann-sheet migration of EPs and the decoupling of physical branch segments}
\label{sec:migration}

The exceptional points that organize avoided crossings in the vacuum spectrum are not static: when fluid loading is applied, they move continuously on the dispersion manifold $\mathcal{M}$. Whether this migration alters the observable spectrum, however, depends fundamentally on the spectral sector in which the EP resides.

Vacuum avoided crossings come in two topological classes, distinguished by the regime of the controlling EP pair. \emph{Trapped} (subsonic) EPs have real $k$ and purely imaginary $k_z$; the fluid impedance is real, the operator is Hermitian for real frequencies, and the pair is confined to $\Im(k)=0$, $\Im(k_z)>0$ for all fluid densities. Such EPs can never cross an admissibility boundary, and the veerings they govern are permanent features of the observable spectrum. \emph{Leaky} (supersonic) EPs have complex $k$ and complex $k_z$; the operator is genuinely non-Hermitian, and the pair can migrate across the admissibility boundaries $\Im(k)=0$ or $\Im(k_z)=0$ as $\rho_f$ increases.

In this section we establish the second mechanism by which fluid loading restructures the physical spectrum---the migration of \emph{leaky} EP pairs out of the physical sector, and the consequent loss of topological coupling between the physical segments of the two branches they previously connected. This mechanism explains why \emph{certain} vacuum avoided crossings are erased under loading---specifically, those whose controlling EPs are leaky---and why, among those, the narrowest veerings disappear first. Trapped EPs and their associated veerings are unaffected by this mechanism and persist for all loading densities.

\subsection{EP trajectories and the physical sector boundary}
\label{sec:EP_trajectories}

We restrict attention to \emph{leaky} EP pairs: those for which, at zero loading, $\Im(k)>0$ and both $k$ and $k_z$ are complex. Trapped EPs are confined to the Hermitian region and do not cross admissibility boundaries; they are excluded from the migration analysis that follows.

With fluid density $\rho_f$ admitted as a free parameter, the EP condition comprises four real equations in the five real variables $(\Re\omega,\Im\omega,\Re k,\Im k,\rho_f)$; its solutions therefore form one‑parameter families. As loading is applied, a leaky EP traces a \emph{continuous trajectory} $(\omega_{\mathrm{EP}}(\rho_f), k_{\mathrm{EP}}(\rho_f), k_{z,\mathrm{EP}}(\rho_f))$ on $\mathcal{M}$. The physical sector $\mathcal{M}_{\mathrm{phys}}$ is bounded by the admissibility hypersurfaces
\begin{equation}
    \mathcal{B}_k = \{ \Im(k)=0 \}, \qquad
    \mathcal{B}_{k_z} = \{ \Im(k_z)=0 \},
    \label{eq:boundaries_migration}
\end{equation}
each of codimension one within the EP trajectory. Consequently, a leaky EP family crosses a boundary at an isolated value of $\rho_f$.

Two independent crossing events are therefore possible along an EP locus:
\begin{enumerate}[label=(\roman*)]
    \item $\Im(k_{\mathrm{EP}})=0$: the EP meets the causality boundary;
    \item $\Im(k_{z,\mathrm{EP}})=0$: the EP meets the radiation boundary, leaving the outgoing sheet.
\end{enumerate}
Either event alone removes the EP pair from $\mathcal{M}_{\mathrm{phys}}$. Once outside the physical sector, the pair no longer lies within the region of the Riemann surface sampled by the physical observation space $\Omega_{\rm obs}$. The two physical segments of the branches, previously coupled by the EP, become independent curves within $\Omega_{\rm obs}$. As a consequence, they may intersect freely on the real frequency axis, with uncorrelated imaginary parts—the observable signature that the EP no longer constrains the two physical segments. This is the topological mechanism of the veering disappearance described in \cref{sec:migration}.

If an EP family reaches $\Im(\omega)=0$ while still inside $\mathcal{M}_{\mathrm{phys}}$, a genuine degeneracy becomes directly observable: the two branches intersect at a real frequency with vertical tangents. This event, however, requires the EP to remain within the physical sector up to the real axis, which is non‑generic; in the examples of \cref{sec:validation}, the EPs that erase vacuum veerings leave the physical sector before their frequencies become real.

\begin{remark}[Simultaneous boundary crossing]
\label{rem:simultaneous_migration}
A simultaneous occurrence of $\Im(k)=0$ and $\Im(k_z)=0$ forces $k$ and $k_z$ to be real; the constraint $k_z^2+k^2=\omega^2/c_f^2$ then requires $\omega$ to be real. Such a point lies on the intersection of the two admissibility boundaries with the real frequency axis, corresponding to an EP that exits $\mathcal{M}_{\mathrm{phys}}$ precisely as it becomes observable. This simultaneous event has codimension two and is not necessary for the veering‑to‑crossing transition.
\end{remark}

\begin{remark}[Why EP exit decouples the physical segments]
\label{rem:EP_exit_geometry}
A global dispersion branch extends continuously across the boundary $\Im(k_z)=0$, connecting its physical segment to its non‑physical continuation without approaching the square‑root branch point $k_z=0$. Two global branches are linked by a conjugate EP pair. When the EP pair resides on the physical sheet, the two physical segments lie on branches that exchange identity upon encircling the EP—they are topologically coupled. When the pair migrates to the non‑physical sheet, the two global branches remain linked by the EP, but the connection now lies outside the physical observation space $\Omega_{\rm obs}$. Within $\Omega_{\rm obs}$, the two physical segments belong to branches that, \emph{as sampled in the observation space}, are no longer forced to exchange identity. They become independent curves and may intersect freely.
\end{remark}

\subsection{The veering-to-crossing transition}
\label{sec:veering_to_crossing}

The consequence of the boundary crossing for the observable spectrum is formalized as follows.

\begin{proposition}[Veering-to-crossing transition by EP exit]
\label{prop:veering_crossing}
Let two physical branches be connected by a conjugate EP pair that resides in $\mathcal{M}_{\mathrm{phys}}$ at zero loading (the vacuum avoided crossing). Then:
\begin{enumerate}[label=(\roman*)]
    \item While the pair remains in $\mathcal{M}_{\mathrm{phys}}$, the two branches exhibit real‑wavenumber veering along the real frequency axis (\cref{prop:signature}).
    \item If, at some loading $\rho_f^*$, the pair crosses $\mathcal{B}_k$ or $\mathcal{B}_{k_z}$ and enters a non‑physical region of $\mathcal{M}$, the two physical branch segments within $\Omega_{\rm obs}$ are no longer connected by any EP. They become independent curves and may intersect freely on the real frequency axis, with uncorrelated imaginary parts. No real‑wavenumber veering occurs, and the two modes exchange neither attenuation nor energy velocity at the intersection—because no topological constraint forces such an exchange.
\end{enumerate}
\end{proposition}

\begin{proof}
Part (i) is exactly \cref{prop:signature}. For part (ii), once the pair leaves $\mathcal{M}_{\mathrm{phys}}$, the EP that previously linked the two global branches lies outside the physical observation space $\Omega_{\rm obs}$. Within $\Omega_{\rm obs}$, the two physical segments are no longer constrained by any degeneracy of the physical eigenvalue problem. They behave as independent analytic curves on the real frequency axis and may intersect with a finite angle; their $\Im(k)$ values are independent because no topological constraint forces them to swap. The coupling that produced the avoided crossing is lifted by the departure of the EP from the region sampled by $\Omega_{\rm obs}$.
\end{proof}

\Cref{prop:veering_crossing} provides the diagnostic already used in \cref{sec:signature}: a true crossing with uncorrelated imaginary parts signals the absence of a physical‑sector EP pair.

\subsection{Selectivity: why narrow-gap veerings vanish first}
\label{sec:selectivity}

Vacuum avoided crossings exhibit a wide range of gap sizes. The gap $\Delta k_{\min} = \min_\omega|\Re(k_+(\omega)-k_-(\omega))|$ is controlled by the distance $|\Im(\omega_{\mathrm{EP}})|$ of the controlling EP pair from the real frequency axis. A narrow veering corresponds to an EP pair that is close to the real axis and, typically, also close to the admissibility boundaries $\mathcal{B}_k$ or $\mathcal{B}_{k_z}$.

When $\rho_f$ is increased, the EP moves. The distance a pair must travel before crossing a boundary is roughly proportional to its initial margin to that boundary. Hence pairs with the smallest margin---those associated with the narrowest vacuum veerings---exit $\mathcal{M}_{\mathrm{phys}}$ at the smallest loading densities. Wide‑gap pairs, lying deep inside the physical sector, either remain within $\mathcal{M}_{\mathrm{phys}}$ for all physically attainable $\rho_f$ (and the associated branch segments remain coupled, preserving the veering) or require much higher densities to cross the boundary.

\begin{proposition}[Loading order of veering disappearance]
\label{prop:order}
Consider a set of vacuum avoided crossings ordered by increasing veering gap $\Delta k_{\min}$. Under monotonic increase of $\rho_f$, the crossings disappear in order of increasing gap: the narrowest veering is erased first, the widest last (or never).
\end{proposition}

\begin{proof}[Heuristic argument]
The initial distance of an EP pair to the nearest admissibility boundary is a continuous function of the pair's coordinates. For a pair close to the real frequency axis, the Puiseux expansion gives $\Delta k_{\min} \propto |\Im(\omega_{\mathrm{EP}})|$. Meanwhile, the proximity to, say, $\mathcal{B}_{k_z}$ can be parameterised by $|\Im(k_z)|$. Because the trajectory in $\rho_f$ is smooth, a pair with a smaller initial margin reaches zero at a smaller $\rho_f$, triggering the transition earlier. A fully quantitative version requires tracking the specific EP trajectory, which is done numerically in \cref{sec:validation}.
\end{proof}

\Cref{prop:order} resolves the selectivity puzzle posed in the Introduction: only the narrowest veerings vanish under weak loading, while wider ones survive. The mechanism is topological, independent of the details of the laminate, the fluid, or the discretization.

\subsection{Numerical tracking of the migration path}
\label{sec:migration_path}

To make the migration concrete, we outline the numerical tracking of an EP along a density homotopy. The augmented system \eqref{eq:augmented_app} is extended by the degeneracy condition (vanishing of the discriminant of the two nearly degenerate eigenvalues), yielding a closed system for $(\omega_{\mathrm{EP}}, k_{\mathrm{EP}}, k_{z,\mathrm{EP}})$ at each $\rho_f$. The solution is initialized from the vacuum EP (located by solving the two‑eigenvalue problem for the analytically continued vacuum operator) and continued with $\rho_f$ as the parameter.

For an EP pair that controls a narrow vacuum veering, the typical migration is as follows. At $\rho_f=0$, both partners lie on the physical sheet with $\Im(k)>0$, $\Im(k_z)>0$ and $\Im(\omega_{\mathrm{EP}})\neq0$ (one partner with $\Re k>0$, the other with $\Re k<0$). As $\rho_f$ increases, $\Im(k_z)$ of the tracked partner decreases monotonically; its conjugate behaves symmetrically. At a critical density $\rho_f^*$, $\Im(k_z)$ reaches zero. Beyond $\rho_f^*$, the pair resides in $\Im(k_z)<0$, and the two physical branch segments within $\Omega_{\rm obs}$ are no longer coupled by the EP and are observed to intersect freely on the real axis. This evolution is illustrated in \cref{fig:unsym_double} of \cref{sec:validation}.

In the subsonic regime, where $k_z$ is purely imaginary, the boundary $\Im(k_z)=0$ coincides with $k_z=0$, i.e.\ the sound line. An EP can therefore leave the physical sector by reaching the sound line; this is observed for trapped modes near the Scholte‑wave limit.

\subsection{Relation to the expansion of the observation space}
\label{sec:relation_to_obs_migration}

The veering-to-crossing transition induced by EP migration is independent of the observation‑space expansion described in \cref{sec:observation_space}. The latter adds new segments to existing branches or reveals emergent modes, whereas the former rewires the connectivity among branches that are already present. The two mechanisms can operate concurrently: a branch extended to the sound line by the first mechanism may undergo a veering‑to‑crossing transition with a neighbour if the controlling EP pair migrates away. The full physical spectrum of a fluid‑loaded waveguide is the joint outcome of both mechanisms.

The numerical validation in \cref{sec:validation} illustrates the two mechanisms separately and in combination, for symmetric and asymmetric laminates under single‑ and double‑sided water loading. EP trajectories are tracked explicitly, the crossing of admissibility boundaries is demonstrated, and the predicted ordering of veering disappearance is confirmed.

\section{Conditional exceptional-point formation under mirror-symmetry breaking}
\label{sec:single_sided}

The two restructuring mechanisms established in \cref{sec:observation_space,sec:migration} concern the fate of EPs that already exist in the vacuum limit. Fluid loading can, however, also \emph{create} new EPs when it breaks the mirror symmetry of the plate. The perturbation-theoretic framework for mode veering in conservative elastic waveguides \cite{xiao_mode_2026} provides the basis for understanding how a symmetry-protected diabolic point in a Hermitian system unfolds when non-Hermitian coupling is introduced. In this section we determine the condition under which a symmetry-protected diabolic point unfolds into an exceptional point, and show that the outcome is not universal: it is generic in the subsonic (trapped) regime, yet conditional in the supersonic (leaky) regime.

\subsection{Symmetry-protected diabolic points in the vacuum plate}
\label{sec:vacuum_dp}

For a plate symmetric about its mid-plane, the vacuum operator commutes with the reflection $z\mapsto -z$. The modal space decomposes into symmetric (S) and antisymmetric (A) Lamb waves, and modes of opposite symmetry do not interact. Their intersections on the real $(\omega,k)$ plane are therefore symmetry-protected diabolic points (DPs)---true crossings where two real eigenvalues coincide without coupling.

When fluid loading is applied, two distinct scenarios arise depending on whether the loading preserves or breaks the mirror symmetry.

\subsection{Symmetry preservation under double-sided loading}
\label{sec:double_sided}

If identical fluids load both surfaces of a mid-plane symmetric plate, the reflection symmetry $z\mapsto -z$ is preserved. The fluid perturbation projected onto the S--A subspace becomes
\begin{equation}
    \mathbf{V}_f^{\text{(double)}} = Z_f
    \begin{pmatrix}
        u_S^2 + (-u_S)^2 & u_S u_A + (-u_S)u_A \\[2pt]
        u_S u_A + (-u_S)u_A & u_A^2 + u_A^2
    \end{pmatrix}
    = 2Z_f
    \begin{pmatrix}
        u_S^2 & 0 \\[2pt]
        0 & u_A^2
    \end{pmatrix},
    \label{eq:Vf_double}
\end{equation}
where the signs follow from the parity of the normal displacements at the loaded surfaces ($u_S$ odd, $u_A$ even). The off-diagonal terms cancel exactly: the S and A families remain decoupled, and their crossings survive as genuine diabolic points.

\begin{proposition}[Symmetry-protected crossings under double-sided loading]
\label{prop:double_sided}
For a mid-plane symmetric plate with identical fluids on both surfaces, S--A diabolic crossings are protected from EP formation. Intra-class EPs (e.g.\ S$_1$--S$_2$, A$_1$--A$_2$) persist and are shifted in the complex $(\omega,k)$ plane by the fluid loading, governing avoided crossings in both the subsonic and supersonic regimes. Fluid loading does not create intra-class EPs \emph{de novo}; it translates existing vacuum EPs and may alter their distance from the real frequency axis, thereby modifying the sharpness of the associated veering.
\end{proposition}

\subsection{Conditional EP formation under single-sided loading}
\label{sec:single_sided_analysis}

When fluid loads only one surface, the mirror symmetry is broken. The question is whether a symmetry-protected S--A diabolic point necessarily unfolds into an exceptional point.

Consider a vacuum crossing $(k_0,\omega_0)$ with degenerate eigenvectors $\psi_S$, $\psi_A$. Expanding the fluid-loaded operator $\mathbf{D} = \mathbf{D}_0 + \mathbf{H}^\mathsf{T}\mathbf{Z}\mathbf{H}$ to first order and projecting onto the subspace spanned by $\{\psi_S,\psi_A\}$ yields the local $2\times 2$ matrix function
\begin{equation}
    \mathbf{F}_{\!SA}(k,\omega) = (\omega-\omega_0)\mathbf{A} + (k-k_0)\mathbf{B} + \mathbf{V}_f(\omega,k),
    \label{eq:local_2x2}
\end{equation}
where $\mathbf{A}$ and $\mathbf{B}$ are the projections of the vacuum first-derivative operators, and $\mathbf{V}_f$ is the fluid perturbation. For single-sided loading,
\begin{equation}
    \mathbf{V}_f = Z_f(\omega,k)
    \begin{pmatrix}
        u_S^2 & u_S u_A \\[2pt]
        u_S u_A & u_A^2
    \end{pmatrix},
    \label{eq:Vf_single}
\end{equation}
with $u_S$, $u_A$ the normal displacements of the vacuum modes evaluated at the loaded surface. Two effects are present: the diagonal terms shift the individual dispersion surfaces, while the off-diagonal term $V_{SA}=Z_f u_S u_A$ breaks the symmetry selection rule and couples the S and A families.

The fate of the DP under this coupling depends critically on the spectral sector---subsonic versus supersonic---because the analytic character of the fluid impedance $Z_f$ differs fundamentally between the two.

\paragraph{Subsonic (trapped) regime.}
For $v_{\rm p}<c_f$, the vertical wavenumber is purely imaginary, $k_z=i\alpha$ with $\alpha>0$, and the impedance $Z_f=-\rho_f\omega^2/\alpha$ is real. The projected operator $\mathbf{F}_{\!SA}$ is therefore Hermitian for real $(\omega,k)$, and its eigenvalues lie on the real axis. The symmetry-protected DP is a real-axis intersection of two real eigenvalues. When the off-diagonal coupling $V_{SA}$ is switched on, the two real sheets are pushed apart on the real axis---an avoided crossing. Because the eigenvalues are analytic functions of the complex variables $(\omega,k)$, an avoided crossing on the real axis generically implies the existence of a pair of exceptional points in the complex plane: the two Riemann sheets, repelled on the real axis, must reconnect somewhere in the complex parameter space. Hence, in the trapped regime, EP formation is \emph{generic}: any infinitesimal symmetry-breaking perturbation lifts the degeneracy and creates a nearby conjugate EP pair. These EPs are \emph{hidden}---they lie off the real frequency axis---but they are physically present as the organizing centers of the veering.

\paragraph{Supersonic (leaky) regime.}
For $v_{\rm p}>c_f$, $k_z$ is complex with $\Im(k_z)>0$, and $Z_f$ is complex-valued. The projected operator is genuinely non-Hermitian: its eigenvalues are complex even for real $\omega$, and no real-axis intersection exists to anchor the analytic continuation. Symmetry breaking removes the protection of the vacuum crossing and enables S--A coupling, but the eigenvalues now trace independent complex trajectories in $(\Re k, \Im k)$ space. An EP requires the coincidence of two complex numbers---i.e.\ the simultaneous equality of both $\Re k$ and $\Im k$ for the two branches. Whether this occurs depends on the magnitude of the fluid-induced complex wavenumber splitting relative to the coupling: the trajectories must be brought to a common point in the complex plane. This is not guaranteed by symmetry breaking alone. EP formation in the leaky regime is therefore \emph{conditional}.

\begin{proposition}[Conditional EP formation under single-sided loading]
\label{prop:conditional_EP}
For a mid-plane symmetric plate under single-sided fluid loading, the unfolding of a symmetry-protected S--A diabolic point into an exceptional point is:
\begin{enumerate}[label=(\roman*)]
    \item \emph{generic} in the subsonic (trapped) regime, where the real impedance makes the perturbed operator Hermitian and any real-axis avoided crossing guarantees a nearby EP pair;
    \item \emph{conditional} in the supersonic (leaky) regime, where the complex impedance renders the operator non-Hermitian and EP formation requires the additional coincidence that the radiation-induced complex wavenumber splitting be sufficiently small for the eigenvalue trajectories to intersect.
\end{enumerate}
\end{proposition}

This asymmetry between the two regimes is a fundamental consequence of the analytic structure of the fluid impedance $Z_f\propto 1/k_z$, which is real for imaginary $k_z$ (subsonic) and complex for complex $k_z$ (supersonic). It implies that the observable real-frequency behavior---whether two interacting branches veer or cross---cannot be predicted from symmetry arguments alone; the sector of the Riemann manifold on which the interaction occurs is decisive.

\subsection{Summary of mode interaction mechanisms}
\label{sec:summary_table}

\Cref{tab:EP_summary} summarizes the contrasting topological outcomes for a mid-plane symmetric plate under the three loading conditions. The central message is that symmetry breaking is a necessary but not sufficient condition for EP creation: the subsonic/supersonic character of the interacting modes determines whether new EPs are born, while the migration of existing EPs (\cref{sec:migration}) determines whether inherited veerings survive.

\begin{table}[t]
  \centering
  \caption{Summary of mode interaction mechanisms for a mid-plane symmetric plate under different loading conditions.}
  \label{tab:EP_summary}
  \begin{tabular}{@{}p{2.5cm}p{4.0cm}p{4cm}p{4cm}@{}}
    \toprule
    & \textbf{Vacuum (Hermitian)} & \textbf{Double-sided loading} & \textbf{Single-sided loading} \\
    \midrule
    S--A coupling 
    & Decoupled (symmetry) 
    & Decoupled (symmetry preserved) 
    & Coupled (symmetry broken) \\
    \addlinespace
    S--A diabolic points 
    & Real crossings (protected) 
    & Protected (no EP) 
    & Unfolded; EP formation generic in trapped regime, conditional in leaky regime \\
    \addlinespace
    Intra-class EPs 
    & Govern avoided crossings on real axis
    & Shifted by radiation loss; govern trapped- and leaky-mode veering 
    & Shifted by radiation loss; additionally influenced by S--A coupling \\
    \bottomrule
  \end{tabular}
\end{table}

\section{Computational inter-manifold transport framework}
\label{sec:computation}

The theoretical mechanisms of \cref{sec:observation_space,sec:migration,sec:single_sided} are implemented through a vacuum-anchored inter-manifold transport strategy. The complete Hermitian vacuum spectrum provides an unambiguous classification of every mode; density homotopy transports selected vacuum seeds onto the fluid-loaded manifold; and adaptive real-$\omega$ continuation traces the physical branches. This section outlines the procedure and its theoretical guarantees, with algorithmic details deferred to \cref{app:homotopy}.

\subsection{Seed selection from the vacuum spectrum}
\label{sec:seeds_computation}

Vacuum modes are classified by their phase velocity relative to the fluid sound speed, as described in \cref{sec:downward_continuation}. For each mode, one or two seeds are placed at frequencies where the vacuum eigenpair is well separated from neighboring branches and far from the sound line. Straddling modes require two seeds: one near the cut-off to capture the downward extension towards the sound line, and one at the upper end of the observation window to capture the supersonic segment. Each seed is initialized on the outgoing sheet by enforcing $\Im(k_z)>0$.

\subsection{Density homotopy}
\label{sec:homotopy_computation}
The inter-manifold transport strategy employed here follows the homotopy-continuation framework established for general viscoelastic waveguides \cite{xiao_homotopy_2026}, adapted here to the fluid-loaded problem with the fluid density as the homotopy parameter. A one-parameter family is introduced by scaling the fluid density,
\begin{equation}
    \rho_f(s)=s\,\rho_{\text{target}},\qquad s\in[0,1].
\end{equation}
The augmented system appends the fluid constraint $k_z^2+k^2=\omega^2/c_f^2$ and treats $k_z$ as an independent unknown, removing the square-root ambiguity during transport. Following a complex path in $s$ avoids the isolated exceptional points, so each seed at $s=1$ inherits the exact physical identity of its vacuum ancestor. The predictor–corrector implementation is given in \cref{app:homotopy}.

\subsection{Real-frequency continuation and topological validation}
\label{sec:continuation_computation}

From each transported seed, the branch is traced along the real frequency axis using $\omega$ as the continuation parameter. For generic physical-sector EP pairs with $\Im\omega_{\mathrm{EP}}\neq0$, the real axis contains no exceptional point, so the continuation never encounters a true degeneracy; it traverses avoided crossings where eigenvectors rotate continuously. The step size is controlled by the Modal Assurance Criterion: a rapid drop in MAC signals proximity to an underlying EP and triggers refinement. Every veering is checked against \cref{prop:signature}; a true crossing with uncorrelated imaginary parts in a branch inherited from a vacuum avoided crossing signals that the controlling EP has left the physical sector, as established in \cref{sec:migration}.

\subsection{Emergent modes and the limits of the method}
\label{sec:emergent_computation}

Branches without a vacuum ancestor, such as the symmetric Scholte wave under certain double-sided loadings, cannot be reached by density homotopy. They are captured by a direct low-frequency search on the fully loaded system, using the Stoneley speed as initial guess. The full physical spectrum within the seeded frequency band is the union of vacuum-ancestor branches and these emergent modes. No finite frequency window can exhaust the infinite-dimensional open-waveguide spectrum; the framework guarantees topological consistency and completeness relative to the chosen vacuum anchors, rather than absolute exhaustiveness.

\section{Validation of the restructuring mechanisms in fluid-loaded laminates}
\label{sec:validation}

The theoretical framework of \cref{sec:observation_space,sec:migration,sec:single_sided} establishes two independent restructuring mechanisms---the expansion of the physical observation space and the migration of exceptional points across admissibility boundaries---as well as the conditional nature of EP formation under mirror-symmetry breaking. This section validates these mechanisms through three sets of numerical experiments on composite laminates. The symmetric laminate under double-sided loading tests symmetry preservation, the emergence of Scholte waves, and the observation-space expansion mechanism. The symmetric laminate under single-sided loading tests the conditional EP formation of \cref{sec:single_sided}. The asymmetric laminate, in which no symmetry-protected crossings exist, is used to isolate and validate the EP-migration mechanism of \cref{sec:migration}, and to further confirm the observation-space expansion.

\subsection{Model configuration}
\label{sec:setup}

The solid is a unidirectional carbon-fiber/epoxy composite (density $\rho_s=1560\,\mathrm{kg/m^3}$, stiffness constants from \cite{castaings_finite_2008}), assembled into two 16-ply laminates of total thickness $h=4\,\mathrm{mm}$:
\begin{itemize}
    \item \textbf{Symmetric laminate}: $[0/90/45/\text{-}45]_{2\mathrm{s}}$, possessing mid-plane reflection symmetry.
    \item \textbf{Asymmetric laminate}: $[0/90/45/\text{-}45]_{4}$, lacking geometric symmetry.
\end{itemize}
The fluid is water ($\rho_f=1000\,\mathrm{kg/m^3}$, $c_f=1500\,\mathrm{m/s}$). The SAFE discretization employs fifth-order GLL spectral elements (two per physical ply, $\sim$480 displacement unknowns). The frequency–thickness product range $0\le fd \le 5\,\mathrm{MHz\cdot mm}$ is chosen as the observation window where physical dispersion modes and modal interactions are identified.

The inter-manifold transport strategy of \cref{sec:computation} is applied: vacuum anchors are classified by phase velocity relative to $c_f$, transported by density homotopy, and traced by real-$\omega$ continuation. Every avoided crossing is validated against the signature of \cref{prop:signature}. For reference, selected curves are compared with the open-source solver Dispersion Calculator (DC) v3.1 \cite{huber_dispersion_2024}, which employs fixed-frequency root searching without topological guidance.

\subsection{Symmetric laminate under double-sided loading: symmetry preservation, observation-space expansion, and Scholte-wave emergence}
\label{sec:val_double}

\Cref{fig:sym_double}(a) displays the vacuum spectrum of the symmetric laminate. Mid-plane reflection symmetry decouples the spectrum into symmetric (S, red) and antisymmetric (A, blue) families, producing numerous symmetry-protected diabolic crossings. The sound line $\omega=c_f k$ partitions each branch into subsonic ($v_{\rm p}<c_f$) and supersonic ($v_{\rm p}>c_f$) segments. Only the fundamental A$_0$ mode straddles this line; all higher-order modes are entirely supersonic.

Under double-sided water loading, \cref{prop:double_sided} predicts that the S and A families remain decoupled, and that their crossings survive as genuine diabolic points. \Cref{fig:sym_double}(d,e) confirms this precisely: S and A branches intersect in the real part with uncorrelated imaginary parts, showing no veering. Within each symmetry family, adjacent modes interact via conjugate EP pairs and exhibit the signature dictated by \cref{prop:signature}. In the supersonic regime, the avoided crossings display real-part repulsion accompanied by a clear exchange of imaginary parts (\cref{fig:sym_double}(d,e), black box); in the subsonic regime, the trapped A$_0$ branch veers away from the leaky A$_0$ branch with $\Im(k)=0$ throughout (\cref{fig:sym_double}(b,c)).

The spectrum also contains modes that have no vacuum counterpart. The symmetric Scholte wave (S$_0$-type), a purely bound interface mode with phase velocity $v_{\rm p}\lesssim c_f$, is retrieved by the direct low-frequency search; it cannot be accessed by density homotopy because no real-$k$ vacuum ancestor exists. More strikingly, three higher-order leaky branches extend downward to the sound line, terminating at $v_{\rm p}=c_f$ rather than the infinite phase velocity expected from the vacuum cut-off (\cref{fig:sym_double}(d)). This is the observational signature of the \emph{observation-space expansion} established in \cref{sec:observation_space}: the fluid loading enlarges the admissible frequency interval of these branches downward from their vacuum cut-off frequencies to the sound line. These branches are precisely the ones predicted by \cref{prop:extension} and absent from the vacuum observation space.

These sound-line-terminating branches are entirely absent from the DC output (\cref{fig:sym_double}(g-i)), which initializes higher-order modes assuming infinite phase velocity at cut-off and lacks the topological guidance to discover their finite-velocity extensions to the sound line. In total, DC misses five physical branches: the three sound-line-terminating branches discussed above and two additional higher-order branches that terminate at finite cut-off frequencies—though it correctly captures the symmetric Scholte wave. The present method recovers all of them, confirming that the recovered branches are not numerical artifacts but the observable consequence of the loading-induced expansion of the physical observation space.

\begin{figure}[!htbp]
\centering
\includegraphics[width=1.03\columnwidth]{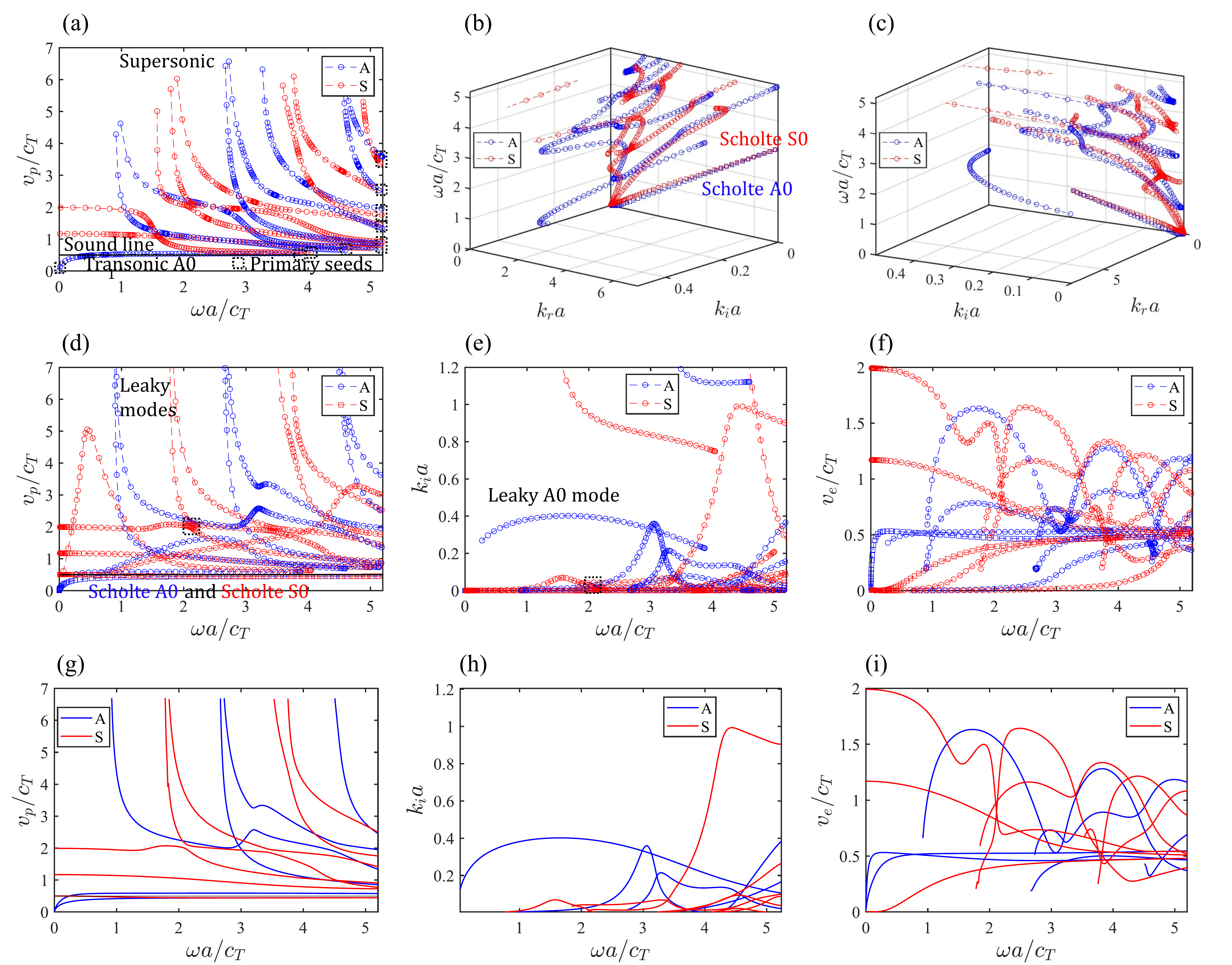}
\caption{
Complete dispersion analysis for the symmetric laminate $[0/90/45/\text{-}45]_{2\mathrm{s}}$ under double-sided water loading.
\textbf{(a)} Vacuum spectrum ($s=0$): symmetric (S, red) and antisymmetric (A, blue) families; the solid black line is the sound line $\omega = k c_f$; black frames mark primary vacuum anchors.
\textbf{(b),(c)} Three-dimensional complex-wavenumber trajectories $(\Re(k),\Im(k),\omega)$ of the fluid-loaded modes, viewed from two angles to expose the evolution along the real frequency axis.
\textbf{(d)} Phase velocity at $s=1$ (present method); arrows mark three higher-order branches extending to the sound line.
\textbf{(e)} Attenuation coefficient $\Im(k)$ at $s=1$; black box highlights a supersonic intra-family avoided crossing with imaginary-part exchange.
\textbf{(f)} Energy flux velocity at $s=1$.
\textbf{(g)} Phase velocity from DC \cite{huber_dispersion_2024}.
\textbf{(h)} Attenuation coefficient from DC; note the omission of the three sound-line-reaching branches.
\textbf{(i)} Energy flux velocity from DC.
}
\label{fig:sym_double}
\end{figure}

\subsection{Symmetric laminate under single-sided loading: conditional exceptional-point formation}
\label{sec:val_single}

When water loads only one surface, the mid-plane reflection symmetry is broken. \Cref{prop:conditional_EP} predicts a fundamental asymmetry between the supersonic and subsonic regimes. We test both limits.

\paragraph{Supersonic (leaky) regime.} With physical water ($c_f=1500\,\mathrm{m/s}$), all higher-order modes are supersonic. \Cref{fig:sym_single}(b--d) reveals the fate of the two symmetry-protected S--A diabolic points marked by red squares: at normalized frequency $\approx 2.89$ and $3.02$ the two modes intersect in $\Re(k)$, yet their imaginary parts differ by more than a factor of 10 (\cref{fig:sym_single}c). The eigenvalues therefore remain far apart in the complex $k$-plane; the simultaneous equality of real and imaginary parts required for an exceptional point is not achieved. This confirms the \emph{conditional} nature of EP formation in the leaky regime: symmetry breaking removes the protection of the vacuum crossing, but the radiation-loss-induced complex wavenumber splitting pushes the two trajectories apart in the imaginary direction, preventing coalescence.

The absence of an exceptional point is further evidenced by the energy-flux response. At each crossing, one mode exhibits a pronounced dip in $v_e$ while the other remains essentially smooth (\cref{fig:sym_single}d). In an EP-mediated avoided crossing, the two modes would exchange physical identities symmetrically---both energy flux and attenuation would undergo a mutual, continuous transfer. The observed asymmetric response indicates that the modes remain independent: the dipping mode undergoes a rapid reconstruction of its radiation mechanism at the crossing frequency, whereas the other mode evolves unaffected. The triad of real-part crossing, uncorrelated imaginary parts, and asymmetric energy-flux response constitutes a definitive diagnostic that no exceptional point governs these interactions.

\begin{figure}[!htbp]
\centering
\includegraphics[width=1.03\columnwidth]{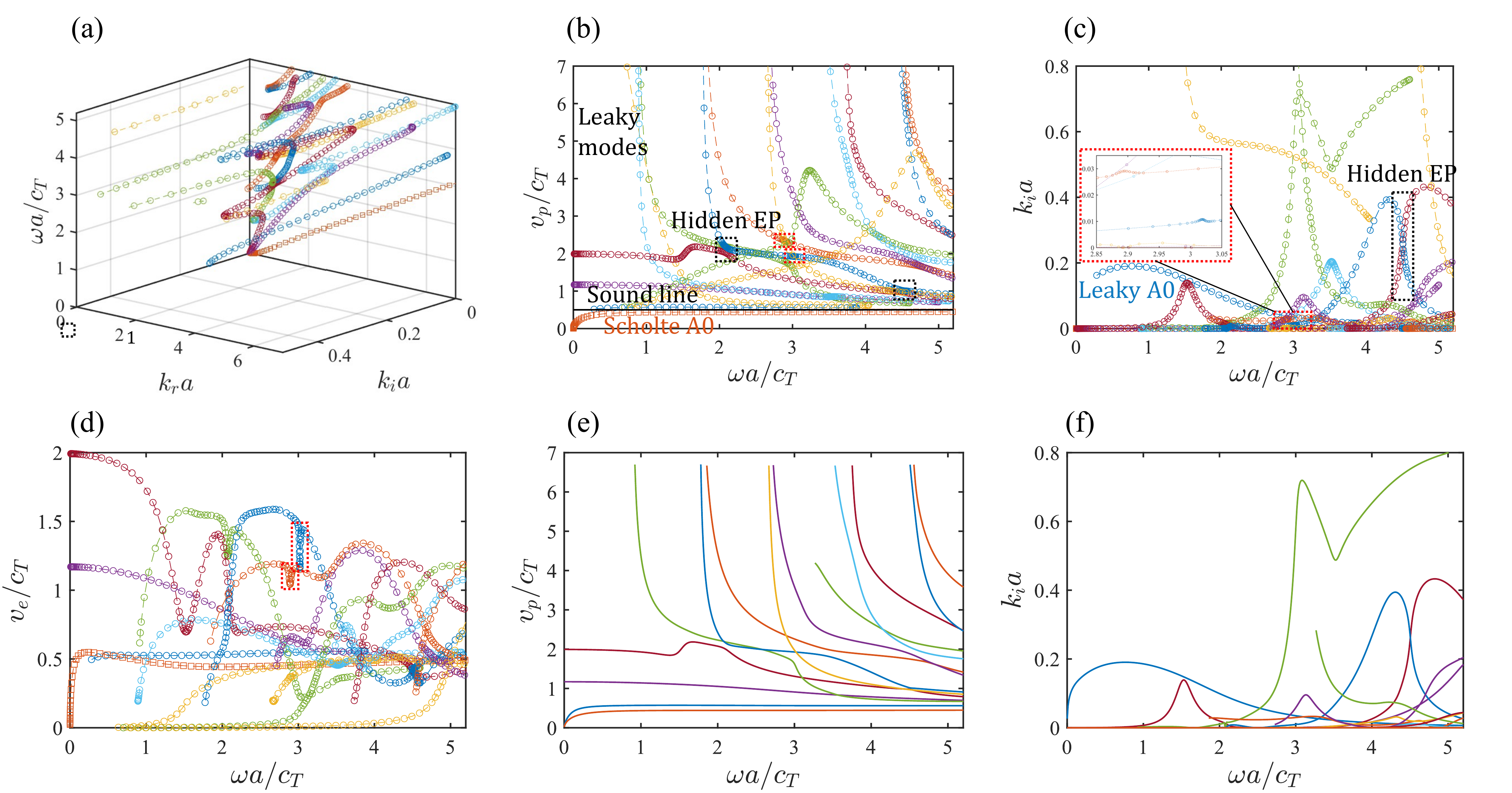}
\caption{
Single-sided water loading of the symmetric laminate $[0/90/45/\text{-}45]_{2\mathrm{s}}$ ($c_f=1500\,\mathrm{m/s}$).
\textbf{(a)} Three-dimensional view of complex-wavenumber trajectories $(\Re(k),\Im(k),\omega)$.
\textbf{(b)} Phase velocity; boxes highlight intra-family avoided crossings governed by EP pairs.
\textbf{(c)} Attenuation coefficient $\Im(k)$; arrow C marks an S--A crossing that persists as a genuine crossing (no EP).
\textbf{(d)} Energy flux velocity.
\textbf{(e)} Phase velocity from DC.
\textbf{(f)} Attenuation coefficient from DC; note the missing branches.
}
\label{fig:sym_single}
\end{figure}

\paragraph{Subsonic (trapped) regime.} To isolate the trapped limit, we introduce a fictitious ``heavy water'' ($\rho_f=1000\,\mathrm{kg/m^3}$, $c_f=25000\,\mathrm{m/s}$) whose sound speed exceeds all vacuum phase velocities. Every mode is then subsonic, the fluid impedance is real, and the system is Hermitian. \Cref{fig:sym_single_alltrapped}(a) shows the phase velocity obtained with a moderate tracking tolerance (MAC $=0.99$). Most S--A crossings have unfolded into avoided crossings, but one intersection (black box) still resembles a diabolic point. When the tolerance is tightened to MAC $=0.999$ (\cref{fig:sym_single_alltrapped}(b)), the adaptive continuation refines the step size sufficiently to resolve the narrow veering. The energy velocity diagram (\cref{fig:sym_single_alltrapped}(c)) confirms the modal exchange: the two branches swap their energy velocities at the veering, a signature of an underlying EP pair.

\begin{figure}[!htbp]
\centering
\includegraphics[width=1.03\columnwidth]{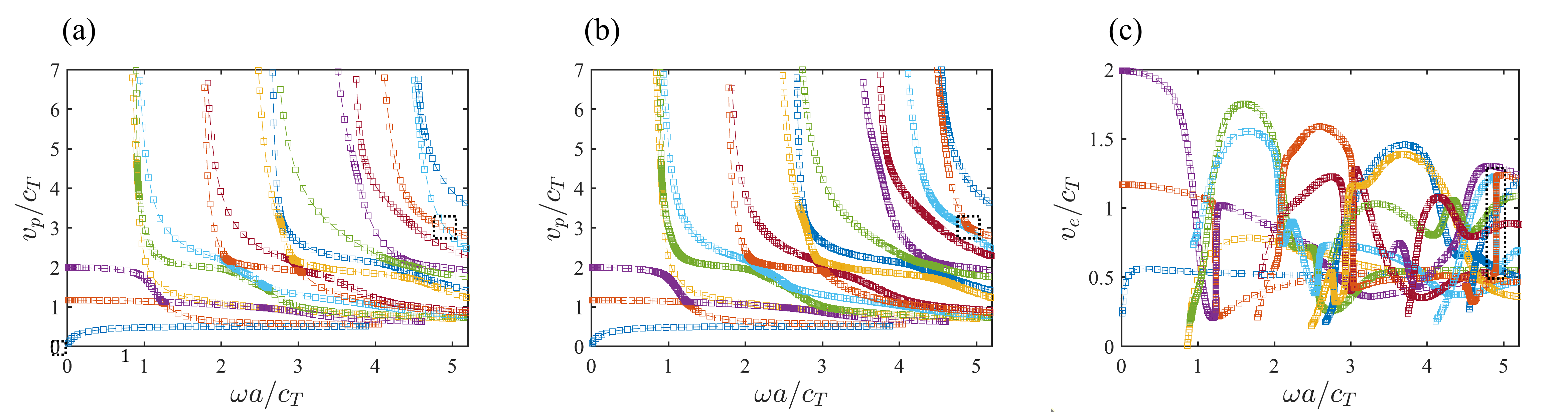}
\caption{
Single-sided ``heavy water'' loading ($c_f=25000\,\mathrm{m/s}$, $\rho_f=1000\,\mathrm{kg/m^3}$). All modes are trapped ($\Im(k)=0$).
\textbf{(a)} Phase velocity with MAC tolerance $0.99$; box D highlights an unresolved S--A crossing.
\textbf{(b)} Phase velocity with MAC tolerance $0.999$; the avoided crossing at D is now resolved.
\textbf{(c)} Energy velocity with MAC tolerance $0.999$; note the exchange of energy velocity at each avoided crossing, signalling EP-mediated modal coupling.
}
\label{fig:sym_single_alltrapped}
\end{figure}

The narrowness of the veering window indicates that the EPs lie extremely close to the real frequency axis in the complex plane---a direct consequence of the weak fluid-induced coupling in the subsonic regime. This validates the \emph{generic} nature of EP formation in the trapped regime: any infinitesimal symmetry-breaking perturbation immediately lifts the degeneracy and creates a nearby EP pair, precisely as the local two-mode model of \cref{sec:single_sided} predicts. Together, the two experiments confirm \cref{prop:conditional_EP}: the same symmetry-breaking perturbation produces qualitatively different spectral topologies depending on the spectral sector.

\subsection{Asymmetric laminate: EP migration and veering-to-crossing transition}
\label{sec:val_asym}

The asymmetric laminate $[0/90/45/\text{-}45]_{4}$ lacks mid-plane symmetry; in vacuum, every close approach between two branches is already an avoided crossing organised by a conjugate EP pair. Under double-sided water loading, the non-Hermitian spectrum inherits this veering structure, but with a key difference: certain narrow-gap avoided crossings disappear and are replaced by true crossings with uncorrelated imaginary parts, while wider veerings persist. This is the signature of the \emph{EP-migration mechanism} established in \cref{sec:migration}.

\Cref{fig:unsym_double}(d,e) shows that among the numerous avoided crossings present in the vacuum spectrum, several narrow-gap veerings are erased after fluid loading, appearing as genuine crossings of $\Re(k)$ with separated $\Im(k)$. These correspond to EP pairs that, under fluid loading, migrated from the physical sheet to the non-physical sheet, crossing either the causality or the radiation admissibility boundary. The remaining wide-gap veerings continue to exhibit the full EP-mediated signature: real-part repulsion accompanied by imaginary-part exchange, as highlighted by black squares in \cref{fig:unsym_double}(d). This selective disappearance directly validates \cref{prop:order}: the narrowest veerings vanish first because their controlling EPs are initially closest to the admissibility boundaries.

Direct evidence for the migration mechanism is provided by an intermediate state at $s=0.1$. \Cref{fig:unsym_double}(g) shows the phase velocity at this reduced fluid density, where four narrow avoided crossings have already turned into real-part crossings, marked by black squares. This indicates that the corresponding EP pairs crossed the admissibility boundary during the early stage of the density homotopy. Panel (h) tracks the trajectories of these EP pairs from $s=0$ to $s=1$, projected onto the plane of real frequency and $\Im(k_z)$; the crossing of the boundary $\Im(k_z)=0$ is clearly visible, confirming that these EPs leave the outgoing sheet and enter the incoming sheet, thereby decoupling the physical branches.

The absence of symmetry-protected crossings in the asymmetric laminate simplifies the topological interpretation: every observable mode interaction is either a persistent EP-mediated veering or a migration-induced crossing. The conjugate EP pairs are shifted in the complex plane by the radiation loss; pairs that remain within the physical sector preserve their veerings, while those that cross a boundary decouple the physical branches. This demonstrates that the inter-manifold transport framework captures both outcomes without relying on structural symmetry to disentangle mode identities.

As in the symmetric case, three higher-order branches extend to the sound line with finite phase velocity (\cref{fig:unsym_double}(d), arrows) and are absent from the DC phase-velocity curve shown in \cref{fig:unsym_double}(i), further confirming the observation-space expansion mechanism. The complex-wavenumber trajectories (\cref{fig:unsym_double}(b,c)) show these branches as isolated trajectories with significantly larger attenuation than neighbouring modes, crossing them in $\Re(k)$ but remaining far apart in $\Im(k)$---confirming that they are not coupled by EPs to the rest of the spectrum.

\begin{figure}[!htbp]
\centering
\includegraphics[width=1.03\columnwidth]{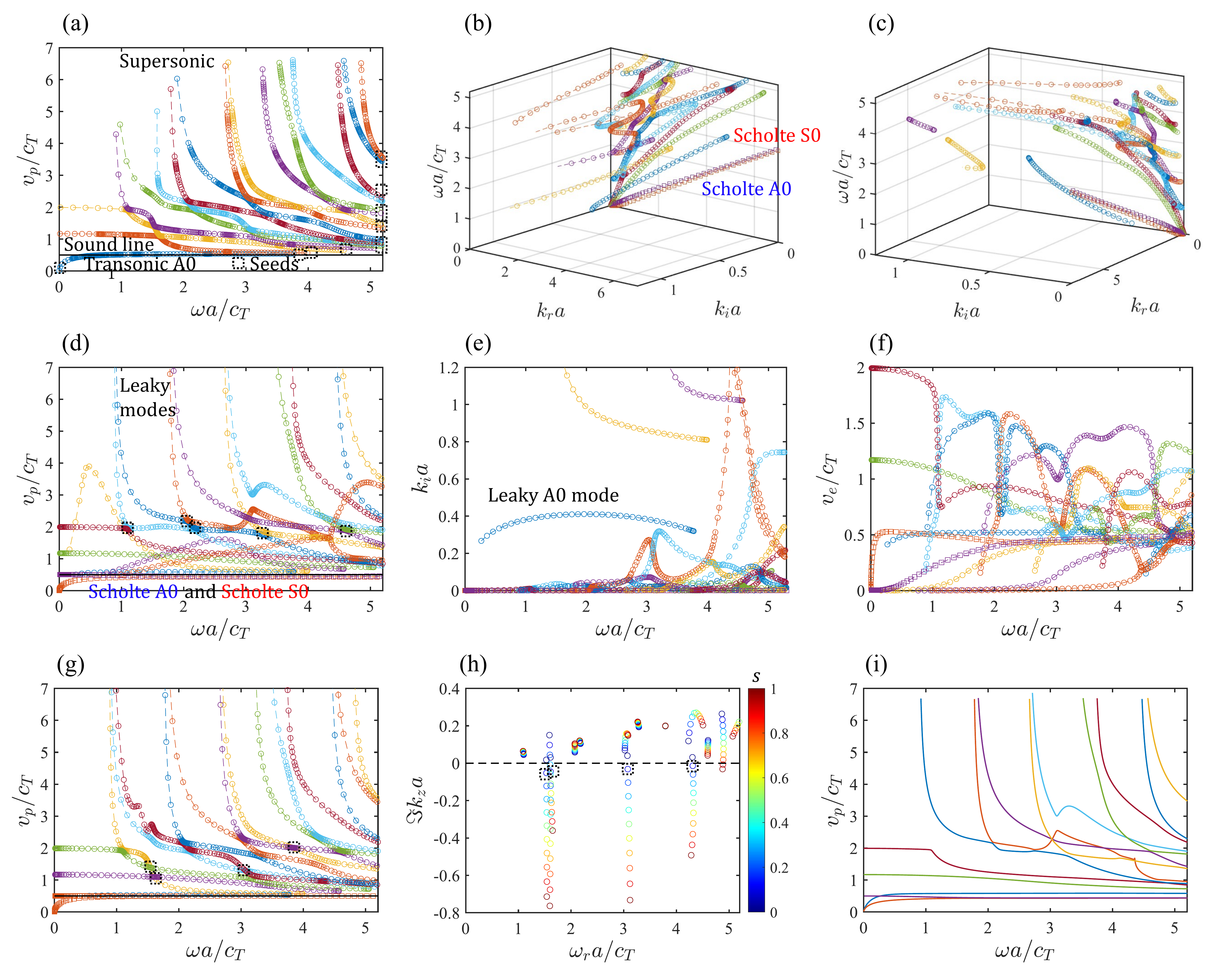}
\caption{
Asymmetric laminate $[0/90/45/\text{-}45]_{4}$ under double-sided water loading.
\textbf{(a)} Vacuum phase velocity ($s=0$); note the absence of symmetry-protected crossings.
\textbf{(b),(c)} Three-dimensional trajectories $(\Re(k),\Im(k),\omega)$ viewed from two angles.
\textbf{(d)} Phase velocity at $s=1$; arrows mark branches extending to the sound line; boxes E and F highlight persistent supersonic avoided crossings with imaginary-part exchange.
\textbf{(e)} Attenuation coefficient $\Im(k)$ at $s=1$.
\textbf{(f)} Energy flux velocity at $s=1$.
\textbf{(g)} Phase velocity at $s=0.1$; four real-part crossings are identified and denoted by black squares, indicating EP shifting from the physical sheet to the non-physical sheet.
\textbf{(h)} EP trajectories from $s=0$ to $s=1$ projected onto the real-frequency versus $\Im(k_z)$ plane; the four EP pairs corresponding to the crossings marked in (g) cross the $\Im(k_z)=0$ boundary, confirming their migration into the incoming sheet.
\textbf{(i)} Phase velocity ($s=1$) obtained with Dispersion Calculator (DC).
}
\label{fig:unsym_double}
\end{figure}

\subsection{Discussion: topological reconstruction of the observable spectrum}
\label{sec:discussion}

The validated mechanisms of \cref{sec:observation_space,sec:migration,sec:single_sided} together establish a coherent picture of fluid-loaded dispersion that goes beyond the individual predictions. In this section we draw out the physical implications that connect these mechanisms and distinguish them from existing views of non-Hermitian waveguides.

\paragraph{Two independent mechanisms, one observable spectrum.}
The physical observation space $\Omega_{\rm obs}$ is shaped by two mutually independent processes. The first---the expansion of the admissible frequency interval---changes \emph{which} solutions are visible: vacuum branches extend below their cut-offs, and emergent interface modes enter the observation window. The second---the migration of exceptional points across admissibility boundaries---changes \emph{how} the visible branches connect: an EP pair that leaves the physical sector decouples the two branches it previously controlled. These two mechanisms are not merely additive; they are jointly responsible for the systematic failure of conventional solvers. A root-finding code that initialises every mode at infinite phase velocity misses the sound-line extensions produced by the first mechanism, while its proximity-based branch sorting misinterprets the decoupled crossings produced by the second. The present framework succeeds because it treats both mechanisms explicitly, rather than attempting to repair their symptoms.

\paragraph{Physical and non-physical sectors are dynamically connected.}
A central lesson of the EP-migration analysis is that the distinction between physical and non-physical sectors is not a fixed property of a mode, but a loading-dependent attribute of the dispersion manifold. A conjugate EP pair that controls a vacuum avoided crossing may, under fluid loading, cross the causality or radiation boundary and pass into the non-physical sector. Once this happens, the corresponding physical branches are no longer connected by the EP, and the avoided crossing inherited from vacuum disappears. This dynamic passage across the admissibility boundary is the topological origin of the selective vanishing of narrow-gap veerings observed in the asymmetric-laminate computations. It also explains why the physical spectrum of a fluid-loaded plate cannot be regarded as a continuous perturbation of the vacuum spectrum: some vacuum connectivities are destroyed, while new connectivities arise from emergent modes.

\paragraph{Regime-dependent EP birth and the role of radiation loss.}
The conditional nature of EP formation under single-sided loading, established in \cref{sec:single_sided} and confirmed numerically in \cref{sec:val_single}, illustrates how radiation loss alters the codimension of spectral degeneracies. In the trapped regime the fluid impedance is real and the problem remains effectively Hermitian; any infinitesimal symmetry-breaking perturbation therefore unfolds a diabolic point into a nearby EP pair, in agreement with the von Neumann--Wigner theorem. In the leaky regime the impedance is complex and the spectrum is genuinely non-Hermitian; the simultaneous equality of two complex eigenvalues imposes an additional condition that is not automatically satisfied. This asymmetry, rooted in the analytic structure $Z\propto 1/k_z$, is not a peculiarity of the chosen laminate or fluid, but a generic consequence of geometric non-Hermiticity. It implies that symmetry arguments alone are insufficient to predict whether two interacting leaky modes will veer or cross.

\paragraph{Towards a general spectral topology of open waveguides.}
The two-mechanism picture developed here for fluid-loaded plates is expected to carry over to a broad class of open mechanical waveguides. Any system coupled to an unbounded radiative medium possesses a radiation sheet, an admissibility boundary, and a physical observation space; its spectrum is therefore subject to the same two mechanisms. In cylindrical shells, layered half-spaces, or acoustic-elastic continua, one should likewise expect vacuum cut-off branches to extend toward the radiation boundary under loading, and EP pairs to migrate across the physical-sector boundaries as the coupling parameter varies. The framework of this work provides the topological language---observation-space expansion, admissibility boundaries, EP migration---in which such predictions can be formulated and tested.

\section{Concluding remarks}
\label{sec:conclusion}

This work has established the exceptional-point mechanisms by which fluid loading reconstructs the guided-wave spectrum of an elastic plate. The central findings are threefold.

First, for real elastic moduli and real fluid parameters, exceptional points on the dispersion manifold occur as forward--backward conjugate pairs, and a physical-sector pair imprints a characteristic signature on the real frequency axis: real-wavenumber veering in both trapped and leaky sectors, with an additional imaginary-wavenumber crossing in the leaky sector. This signature provides a rigorous topological criterion for mode identification.

Second, fluid loading restructures the physical observation space through two independent mechanisms. The observation space expands continuously: branches that in vacuum exist only above a finite cut-off frequency can extend downward to the sound line, populating a frequency band in which the vacuum plate supports no propagating waves. Independently, exceptional points migrate under density variation; when a conjugate pair crosses an admissibility boundary and leaves the physical sector, the avoided crossing it controlled is extinguished, and the two branches decouple, crossing with uncorrelated imaginary parts. The narrowest avoided crossings vanish first, explaining the selective disappearance of small-gap veerings under weak loading.

Third, the birth of new exceptional points under mirror-symmetry breaking is sector-dependent. In the trapped regime, symmetry breaking generically produces EP pairs; in the leaky regime, the complex radiation loss elevates the codimension of the degeneracy, making EP formation conditional. This distinction is a direct consequence of the analytic structure of the fluid impedance and refines the conventional correspondence between symmetry breaking and avoided crossing in non-Hermitian systems.

These findings have been translated into a vacuum-anchored computational framework that combines density homotopy, real-frequency continuation, and topological validation. The framework recovers branches that conventional solvers systematically miss, including sound-line-terminating leaky branches and emergent Scholte waves, and it provides an internal consistency check against spurious crossings.

Several extensions follow naturally. The inclusion of viscoelastic damping or complex fluid properties would break the conjugate EP symmetry and require a broader symmetry classification. The possible coincidence of higher-order exceptional points with the impedance singular locus $k_z=0$ defines a boundary-value problem for the dispersion manifold that remains unexplored. Finally, the two-mechanism picture---observation-space expansion and EP migration across admissibility boundaries---should be tested in other open waveguides, such as cylindrical shells, layered media, and coupled acoustic-elastic continua, where it offers a unifying topological description of spectral reconstruction under radiation coupling.

\appendix
\section{Semi-analytical finite-element discretisation}
\label[appendix]{app:SAFE}

The plate is discretised along the thickness by $N_e$ Gauss--Lobatto--Legendre (GLL) spectral elements with $N_n$ unique nodal points. Three translational degrees of freedom $(q_x, q_y, q_z)$ are retained at each node, collected in the global displacement vector $\mathbf{q}\in\mathbb{R}^{3N_n}$. The plane-harmonic wave ansatz
\begin{equation}
    \mathbf{u}(x,z,t) = \mathbf{N}(z)\,\mathbf{q}\,e^{i(kx-\omega t)}
\end{equation}
is substituted into the principle of virtual work. Standard SAFE assembly \cite{gavric_computation_1995,castaings_finite_2008} yields the vacuum dynamic matrix
\begin{equation}
    \mathbf{D}_0(\omega,k) = \mathbf{K}_1 + ik\mathbf{K}_2 + k^2\mathbf{K}_3 - \omega^2\mathbf{M},
    \label{eq:D0_app}
\end{equation}
where $\mathbf{K}_1,\mathbf{K}_3,\mathbf{M}\in\mathbb{R}^{3N_n\times 3N_n}$ are real symmetric, and $\mathbf{K}_2\in\mathbb{R}^{3N_n\times 3N_n}$ is real skew-symmetric. The explicit entries are
\begin{align}
    \mathbf{K}_1 &= \int_{-h/2}^{h/2} \mathbf{N}'(z)^\mathsf{T} \mathbf{C}_{zz} \mathbf{N}'(z)\,dz, \\
    \mathbf{K}_2 &= \int_{-h/2}^{h/2} \bigl[\mathbf{N}(z)^\mathsf{T} \mathbf{C}_{xz} \mathbf{N}'(z) - \mathbf{N}'(z)^\mathsf{T} \mathbf{C}_{zx} \mathbf{N}(z)\bigr]\,dz, \\
    \mathbf{K}_3 &= \int_{-h/2}^{h/2} \mathbf{N}(z)^\mathsf{T} \mathbf{C}_{xx} \mathbf{N}(z)\,dz, \\
    \mathbf{M} &= \int_{-h/2}^{h/2} \rho_s\,\mathbf{N}(z)^\mathsf{T}\mathbf{N}(z)\,dz,
\end{align}
with $\mathbf{C}_{xx},\mathbf{C}_{zz},\mathbf{C}_{xz},\mathbf{C}_{zx}$ the appropriate sub-blocks of the Voigt stiffness matrix $\mathbf{C}$.

The fluid-loading matrix is obtained from the acoustic impedance \eqref{eq:impedance} via the virtual work of normal traction on the fluid--solid interfaces. Let $i_{\text{top}}$ and $i_{\text{bot}}$ denote the global indices of the normal ($z$) degree of freedom at the top and bottom surfaces. The projection matrix $\mathbf{H}\in\mathbb{R}^{2\times 3N_n}$ extracts these components:
\begin{equation}
    \mathbf{H}_{1,i_{\text{top}}} = 1, \quad \mathbf{H}_{2,i_{\text{bot}}} = 1, \quad \text{all other entries zero}.
\end{equation}
The impedance matrix is diagonal,
\begin{equation}
    \mathbf{Z}(\omega,k_z) = -\frac{i\rho_f\omega^2}{k_z}
    \begin{pmatrix}
        \delta_{\text{upper}} & 0 \\[2pt]
        0 & \delta_{\text{lower}}
    \end{pmatrix},
    \label{eq:Z_app}
\end{equation}
with $\delta_{\text{upper}},\delta_{\text{lower}}\in\{0,1\}$ indicating whether the respective surface is fluid-loaded. The full fluid-loaded dynamic matrix is then
\begin{equation}
    \mathbf{D}(\omega,k,k_z) = \mathbf{D}_0(\omega,k) + \mathbf{H}^\mathsf{T}\mathbf{Z}(\omega,k_z)\mathbf{H}.
    \label{eq:D_full_app}
\end{equation}

The properties used in \cref{sec:EP_topology} follow directly from \eqref{eq:D0_app}--\eqref{eq:D_full_app}: $\mathbf{K}_1^*=\mathbf{K}_1$, $\mathbf{K}_2^*=-\mathbf{K}_2$ (skew-symmetry), $\mathbf{K}_3^*=\mathbf{K}_3$, $\mathbf{M}^*=\mathbf{M}$, and $\mathbf{Z}(\omega^*,-k_z^*)=\mathbf{Z}(\omega,k_z)^*$ on the outgoing sheet, whence \eqref{eq:D_conjugate} holds for the discretised operator.

\section{Homotopy continuation and frequency continuation}
\label[appendix]{app:homotopy}

Both the density homotopy ($s$ as parameter) and the frequency continuation ($\omega$ as parameter) are implemented as predictor--corrector schemes applied to the same augmented system structure. The unknown vector is
\begin{equation}
    \mathbf{y} = [\mathbf{q}^\intercal,\, k,\, k_z]^\intercal \in \mathbb{C}^{n+2},
    \label{eq:y_vector_app}
\end{equation}
where $n=3N_n$ is the number of displacement degrees of freedom.

\subsection{Unified augmented system}
\label{app:augmented}

For a fixed parameter $\lambda$ (either $s$ or $\omega$), the square system $\mathbf{G}(\mathbf{y};\lambda)=\mathbf{0}$ of size $n+2$ reads
\begin{subequations}
\label{eq:augmented_app}
\begin{align}
    \mathbf{D}(\omega,k,k_z;\lambda)\,\mathbf{q} &= \mathbf{0}, \label{eq:aug_D_app} \\
    \mathcal{C}(\omega,k,k_z) &:= k_z^2 + k^2 - \frac{\omega^2}{c_f^2} = 0, \label{eq:aug_C_app} \\
    \mathcal{N}(\mathbf{q}) &:= \mathbf{q}_{\mathrm{ref}}^{\dagger} \mathbf{q} - 1 = 0, \label{eq:aug_N_app}
\end{align}
\end{subequations}
where $\mathbf{q}_{\mathrm{ref}}$ is a constant reference vector (typically the eigenvector from the previous converged point) that removes the eigenvector phase freedom. The dynamic matrix takes the form
\begin{equation}
    \mathbf{D}(\omega,k,k_z;\lambda) = \mathbf{K}_1 + i k\mathbf{K}_2 + k^2\mathbf{K}_3
    + s(\lambda)\,\mathbf{H}^\mathsf{T}\mathbf{Z}(\omega,k_z)\mathbf{H} - \omega(\lambda)^2\mathbf{M},
    \label{eq:D_lambda_app}
\end{equation}
with the understanding that $s(\lambda)=\lambda$ and $\omega(\lambda)=\omega_{\rm fix}$ during homotopy, whereas $s(\lambda)=1$ and $\omega(\lambda)=\lambda$ during frequency continuation.

\subsection{Jacobian and parameter derivatives}
\label{app:jacobian}

The Jacobian with respect to the unknowns $\mathbf{y}$ is the $(n+2)\times(n+2)$ block matrix
\begin{equation}
    \mathbf{G}_{\mathbf{y}}(\mathbf{y};\lambda) =
    \begin{bmatrix}
        \mathbf{D} & \dfrac{\partial\mathbf{D}}{\partial k}\,\mathbf{q} & \dfrac{\partial\mathbf{D}}{\partial k_z}\,\mathbf{q} \\[8pt]
        \mathbf{0}^\intercal & 2k & 2k_z \\[4pt]
        \mathbf{q}_{\mathrm{ref}}^{\dagger} & 0 & 0
    \end{bmatrix},
    \label{eq:G_y_app}
\end{equation}
where
\begin{align}
    \frac{\partial\mathbf{D}}{\partial k} &= i\mathbf{K}_2 + 2k\mathbf{K}_3, \label{eq:D_k_app} \\[4pt]
    \frac{\partial\mathbf{D}}{\partial k_z} &= s(\lambda)\,\mathbf{H}^\mathsf{T}\,\frac{\partial\mathbf{Z}}{\partial k_z}\,\mathbf{H}, \qquad
    \frac{\partial\mathbf{Z}}{\partial k_z} = -\frac{1}{k_z}\,\mathbf{Z}. \label{eq:D_kz_app}
\end{align}

The parameter derivatives are
\begin{equation}
    \frac{\partial\mathbf{G}}{\partial s} =
    \begin{bmatrix}
        \mathbf{H}^\mathsf{T} \mathbf{Z}(\omega,k_z) \mathbf{H} \,\mathbf{q} \\[4pt] 0 \\[4pt] 0
    \end{bmatrix},
    \qquad
    \frac{\partial\mathbf{G}}{\partial\omega} =
    \begin{bmatrix}
        \dfrac{\partial\mathbf{D}}{\partial\omega}\,\mathbf{q} \\[8pt]
        -\dfrac{2\omega}{c_f^2} \\[6pt]
        0
    \end{bmatrix},
    \label{eq:G_param_app}
\end{equation}
with
\begin{equation}
    \frac{\partial\mathbf{D}}{\partial\omega} = \mathbf{H}^\mathsf{T}\frac{\partial\mathbf{Z}}{\partial\omega}\mathbf{H} - 2\omega\mathbf{M},
    \qquad
    \frac{\partial\mathbf{Z}}{\partial\omega} = \frac{2}{\omega}\,\mathbf{Z}.
    \label{eq:dZdw_app}
\end{equation}
In \eqref{eq:dZdw_app}, $k_z$ is treated as an independent variable (consistent with the augmented formulation), so the derivative of $Z\propto\omega^2/k_z$ with respect to $\omega$ contains only the explicit dependence.

\subsection{Homotopy continuation (density parameter $s$)}
\label{app:homotopy_detail}

The homotopy parameter is parametrised by a complex path
\begin{equation}
    s(t) = t + i\,\gamma\,t(1-t), \qquad t\in[0,1],
    \label{eq:complex_path_app}
\end{equation}
with $\gamma\in\mathbb{C}$. The path starts with $\gamma=0$ (real interval). If the tangent predictor or the Newton corrector fails to converge, $\gamma$ is incremented by $0.1i$ and the path is restarted from $t=0$. Since EPs are isolated in the complex $(s,k)$ plane, one or two restarts typically suffice to reach $s=1$.

\paragraph{Tangent predictor.} The tangent vector $\dot{\mathbf{y}} = d\mathbf{y}/dt$ is obtained from
\begin{equation}
    \mathbf{G}_{\mathbf{y}}(\mathbf{y}_j; s(t_j)) \, \dot{\mathbf{y}}_j = -\,\mathbf{G}_s(\mathbf{y}_j; s(t_j)) \, s'(t_j),
    \label{eq:tangent_homotopy_app}
\end{equation}
with $s'(t)=1+i\gamma(1-2t)$. The predictor step is
\begin{equation}
    \mathbf{y}_{j+1}^{(0)} = \mathbf{y}_j + \Delta t\,\dot{\mathbf{y}}_j.
    \label{eq:predictor_app}
\end{equation}

\paragraph{Newton corrector.} The corrector iterates
\begin{equation}
    \mathbf{G}_{\mathbf{y}}(\mathbf{y}_{j+1}^{(m)}; s(t_{j+1})) \, \delta\mathbf{y}^{(m)} = -\,\mathbf{G}(\mathbf{y}_{j+1}^{(m)}; s(t_{j+1})),
    \qquad
    \mathbf{y}_{j+1}^{(m+1)} = \mathbf{y}_{j+1}^{(m)} + \delta\mathbf{y}^{(m)},
    \label{eq:newton_app}
\end{equation}
until $\|\delta\mathbf{y}^{(m)}\|/\|\mathbf{y}_{j+1}^{(m)}\| < \varepsilon_{\rm tol}$ (typically $\varepsilon_{\rm tol}=10^{-6}$). The step size $\Delta t$ is halved upon corrector failure and doubled after two consecutive successful steps, bounded by $\Delta t_{\max}=0.1$.

\paragraph{Algorithmic parameters.} The nominal settings used in this work are: initial step size $\Delta t_0=0.01$, minimum step size $\Delta t_{\min}=10^{-8}$, maximum predictor--corrector iterations $N_{\rm max}=20$, and complex-path increment $\Delta\gamma=0.1i$.

\subsection{Frequency continuation (parameter $\omega$)}
\label{app:freq_cont}

At $s=1$, the dispersion branch is traced along the real frequency axis with $\omega$ as the continuation parameter. The tangent system is
\begin{equation}
    \mathbf{G}_{\mathbf{y}}(\mathbf{y}_j; \omega_j) \, \dot{\mathbf{y}}_j = -\,\frac{\partial\mathbf{G}}{\partial\omega}(\mathbf{y}_j; \omega_j),
    \label{eq:tangent_omega_app}
\end{equation}
and the predictor--corrector structure is identical to \eqref{eq:predictor_app}--\eqref{eq:newton_app}, with $\omega$ replacing $s(t)$.

\paragraph{Adaptive step size.} The step size $\Delta\omega$ is controlled by two independent criteria:
\begin{enumerate}
    \item \emph{Corrector convergence:} $\Delta\omega$ is increased by factor $1.5$ (up to $\Delta\omega_{\max}$) after a step that converges in $N_{\rm iter}<5$ iterations; it is halved when $N_{\rm iter}>15$ or the corrector diverges.
    \item \emph{Eigenvector rotation:} The Modal Assurance Criterion between successive eigenvectors,
    \begin{equation}
        \mathrm{MAC} = \frac{|\mathbf{q}_j^{\dagger} \mathbf{q}_{j+1}|^2}{\|\mathbf{q}_j\|^2 \|\mathbf{q}_{j+1}\|^2},
        \label{eq:MAC_app}
    \end{equation}
    is monitored. When $\mathrm{MAC}<\mathrm{MAC}_{\min}$ (nominal value $0.99$), the eigenvector is rotating rapidly, signalling proximity to an underlying exceptional point in the complex plane; $\Delta\omega$ is halved. 
\end{enumerate}
The nominal settings are: $\Delta\omega_0=10^{-3}\,\omega_{\rm ref}$ (with $\omega_{\rm ref}$ a characteristic plate frequency), $\Delta\omega_{\min}=10^{-8}\,\omega_{\rm ref}$, $\Delta\omega_{\max}=10^{-1}\,\omega_{\rm ref}$, and $\varepsilon_{\rm tol}=10^{-6}$.

\paragraph{Topological validation.} After continuation, every avoided crossing is checked against the signature of \cref{prop:EPpair}. If a suspicious crossing is detected (e.g., a real-part crossing without the required imaginary-part exchange in the supersonic regime), the local interval is re-traced with $\mathrm{MAC}_{\min}=0.999$ to resolve the veering accurately.

\section*{CRediT authorship contribution statement}
\textbf{Dong Xiao}: Conceptualization, Methodology, Software, Data curation, Formal analysis, Writing - Original draft preparation, Writing - Review Editing, Visualization. \textbf{Zahra Sharif-Khodaei}: Supervision, Writing - Review Editing. \textbf{M.H. Aliabadi}: Supervision, Writing - Review Editing.

\section*{Declaration of competing interest}
The authors declare that they have no known competing financial interests or personal relationships that could have appeared to influence the work reported in this paper.

\section*{Acknowledgements}
The first author acknowledges the financial support from the K. C. Wong Postdoctoral Fellowship, funded by the K. C. Wong Education Foundation.

\section*{Data availability}
The source code and data supporting this study will be made publicly available upon publication at \url{https://github.com/dongxiao96/TopoDisper}. 


\setstretch{1.0}
\small{

}

\end{document}